\documentclass[fontsize=11pt,a4paper,DIV=12]{scrartcl}

\usepackage[utf8]{inputenc}
\usepackage[T1]{fontenc}

\usepackage{graphicx}
\usepackage{xcolor}
\usepackage{hyperref}
\usepackage{subcaption}
\usepackage{amsmath}
\usepackage{amssymb}
\usepackage{amsthm}
\usepackage[ruled,noend]{algorithm2e}
\usepackage{booktabs}
\usepackage{longtable}

\usepackage{thmtools}
\usepackage{thm-restate}

\newtheorem{theorem}{Theorem}
\newtheorem{corollary}{Corollary}
\newtheorem{lemma}{Lemma}
\newtheorem{proposition}{Proposition}
\newtheorem{question}{Question}

\SetKwComment{Comment}{/* }{ */}

\def\REP#1+#2!{[#1|#2]}
\def\dd{\delta}
\def\Ret{\vskip2pt\textbf{return} }

\def\BB{\mathcal{B}}
\def\MM{\mathcal{M}}
\DeclareMathOperator{\Prob}{Prob}

\DeclareRobustCommand*{\ora}{\overrightarrow}

\def\inst#1{$^{#1}$}

\newcommand\restr[2]{{
  \left.\kern-\nulldelimiterspace 
  #1 
  \littletaller 
  \right|_{#2} 
  }}

\newcommand{\littletaller}{\mathchoice{\vphantom{\big|}}{}{}{}}

\title{Block coupling and rapidly mixing $k$-heights}
\date{}
\author{Stefan Felsner\inst{1} \and Daniel Heldt\inst{2} \and Sandro Roch\inst{1} \and Peter Winkler\inst{3}}

\begin{document}

\maketitle

\begin{center}
\vskip-12mm
{\footnotesize
\inst{1} 
Institut f\"ur Mathematik, 
Technische Universit\"at Berlin, Germany
\medskip

\inst{2} 
helis GmbH, Dortmund, Germany
\medskip

\inst{3} 
Dartmouth College, Hanover, NH, USA
}
\end{center}

\begin{abstract}
  A~\emph{$k$-height} on a graph~$G=(V, E)$ is an
  assignment~\mbox{$V\to\{0, \ldots, k\}$} such that the value on ajacent vertices
  differs by at most $1$. We study the Markov chain on $k$-heights that in
  each step selects a vertex at random, and, if admissible, increases or
  decreases the value at this vertex by one. In the cases \mbox{of~$2$-heights} and
  $3$-heights we show that this Markov chain is rapidly mixing on certain
  families of grid-like graphs and on planar cubic~$3$-connected graphs.

  The result is based on a novel technique called \textit{block
    coupling}, which is derived from the well-established monotone coupling
  approach. This technique may also be effective when analyzing 
  other Markov chains that operate on configurations of spin systems that form a distributive lattice. It is therefore of independent interest.
\end{abstract}

\section{Introduction}

Markov chains are a generic approach to randomly sample an element from a
collection of combinatorial objects. Examples where Markov chains have been
proven useful include linear extensions
\cite{KaKa91,felsnerWernisch97,BuDy99,DBWi04,huber2006}, eulerian orientations
\cite{winkler96}, graph colorings \cite{friezeVigoda07} and chambers in
hyperplane arrangements \cite{brown98}. See \cite{kannan94} or \cite{LPW17} for more
examples. One is usually interested in \textit{rapidly mixing} Markov chains,
i.e., Markov chains which converge fast towards their stationary
distribution. Here we study a natural Markov chain $\MM$ that operates
on the set of so called \textit{$k$-heights} of some fixed graph~$G$. We
identify a condition on a family $\mathcal{G}$ of graphs which implies
that $\MM$ is rapidly mixing on graphs in $\mathcal{G}$.

For a fixed graph $G=(V, E)$ and a fixed integral upper bound
$k\in\mathbb{N}$, a \textit{$k$-height} is an assignment
$\varphi: V\to\{0, \ldots, k\}$ such that
$\lvert \varphi(v) - \varphi(w)\rvert\leq 1$ for every edge $\{v,w\}\in
E$. See Figure \ref{fig:example_2_height} for an example. In the case in which
$G$ is a grid-like planar graph one may consider a $k$-height as height values
of grid points or as a landscape satisfying a certain smoothness
condition. The set of $k$-heights can also be described as feasible
configurations of a certain \textit{spin system} with uniform weights; in
Section~\ref{sec:glauber_dynamics} we discuss this connection.

\begin{figure}[htb]
\centering
\includegraphics{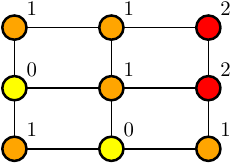}
\caption{Example of a $2$-height. Colors visualize the values as a heat map.}
\label{fig:example_2_height}
\end{figure}

\subsection{Obtaining \texorpdfstring{$k$-heights}{k-heights}
      from \texorpdfstring{$\alpha$-orientations}{alpha-orientations}}

A motivation for studying $k$-heights stems from their connection to
\mbox{\textit{$\alpha$-orientations}} as studied in~\cite{felsner04}. Given an embedded plane graph~$G=(V, E)$ and a
mapping~$\alpha: V\to \mathbb{N}$, an \textit{$\alpha$-orientation} of $G$ is
an orientation of $E$ in which each vertex $v\in V$ has out-degree
$\operatorname{outdeg}(v) = \alpha(v)$. In the example shown in
Figure~\ref{fig:example_alpha_orientation} we have $\alpha(v):=2$ for all
$v\in V$. 

The set of all $\alpha$-orientations of a planar graph $G$ forms a distributive
lattice. There is a unique minimum $\alpha$-orientation
$\overrightarrow{E}_{\text{min}}$ in which there is no bounded face whose
bounding edges form a counterclockwise oriented
cycle. The minimum $\overrightarrow{E}_{\text{min}}$ can be reached from every
$\alpha$-orientation by a sequence of \textit{flips}. A flip\footnote{Here we
  implicitly assume that $\alpha$ and $G$ are such that there are no rigid edges.
  Without this assumption it may be necessary to flip non-facial cycles,
  see~\cite{felsner04}.}
reverts the counterclockwise oriented bounding edges of a face into clockwise
orientation, see the example in Figure~\ref{fig:example_alpha_orientation}.

\begin{figure}[htb]
    \centering
    \begin{subfigure}[b]{.32\textwidth}
	\centering
	\includegraphics[page=1]{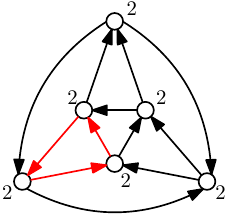}
	\caption{}
	\label{fig:example_alpha_orientation_a}
    \end{subfigure}
    \hfill
    \begin{subfigure}[b]{.32\textwidth}
        \centering
	\includegraphics[page=2]{figures/example_alpha_orientation.pdf}
	\caption{}
	\label{fig:example_alpha_orientation_b}
    \end{subfigure}
    \hfill
    \begin{subfigure}[b]{.32\textwidth}
	\centering
	\includegraphics[page=3]{figures/example_alpha_orientation.pdf}
	\caption{}
	\label{fig:example_alpha_orientation_c}
    \end{subfigure}		
    \caption{(a) Example of an $\alpha$-orientation $\ora{E}_0$. (b)
      the $\alpha$-orientation obtained by flipping the face bounded by
      the red arcs. (c) Minimal
      \mbox{$\alpha$-orientation}~$\ora{E}_{\text{min}}$. Red
      numbers record the number of face flips needed to reach $\ora{E}_0$.}
    \label{fig:example_alpha_orientation}
\end{figure}

The numbers of flips on each bounded face that are necessary to reach
$\overrightarrow{E}_{\text{min}}$ from $\overrightarrow{E}_0$ are independent
of the choice of the flip sequence and uniquely determine
$\overrightarrow{E}_0$. For a sufficiently large $k\in\mathbb{N}$, they form a
$k$-height of the dual graph of~$G$, in which the vertex corresponding to the
unbounded face has value $0$; see
Figure~\ref{fig:example_alpha_orientation_c}.

\subsection{Rapidly mixing Markov chains on
  \texorpdfstring{$k$-heights}{k-heights}}

The state space of the \textit{up/down Markov chain} $\MM$ is the set
of $k$-heights of $G$. A transition of~$\MM$ depends on a vertex
$\Tilde{v}\in V$ and a direction $\triangle\in\{-1, +1\}$; both are chosen uniformly at
random. According to $\triangle$, the value at $\Tilde{v}$ is decremented or
incremented. If the result is again a valid $k$-height, it is accepted as next
state; otherwise, $\MM$ stays at the same state. It is easy to see
that the up/down Markov chain $\MM$ is aperiodic, irreducible and
symmetric, hence it converges towards the uniform distribution on the set of
$k$-heights of $G$.

We could not prove that $\MM$ is rapidly mixing via a direct
application of one of the standard methods such as coupling or canonical paths
(cf. \cite{guruswami2016}). Therefore, we introduce an auxiliary Markov chain
$\MM_\BB$. It depends on some family
$\BB\subset\mathcal{P}(V)$ of \textit{blocks}. A block is a set of
vertices; together they cover $V$, i.e., each vertex $v\in V$ is contained in
at least one block $B\in\BB$. The chain $\MM_\BB$ also operates
on the set of $k$-heights, but in each transition it does not just alter the
value at a single vertex. Instead, it resamples the assgined values on an
entire block at once.

A transition of $\MM_\BB$ can be simulated by a sequence of
transitions of $\MM$. Therefore, the \emph{comparison theorem} of
Randall and Tetali (2000) can be used to transform an upper bound on the
mixing time of $\MM_\BB$ to an upper bound on the mixing
time of~$\MM$. The idea to study a Markov chain $\MM$ by first
studying a boosted chain $\MM_{\BB}$ that performs block moves
rather than single moves is also known in the context of Glauber dynamics
under the term \textit{block dynamics}; see \cite{dobrushinShlosman85, martinelli99, weitz05}.

Our main result is an upper bound on the mixing time of $\MM$ that
depends on a careful choice of~$\BB$. In the statement of the theorem
we let $\partial B$ denote the \textit{boundary} of a block $B\in\BB$,
i.e., the set of vertices outside of $B$ that are adjacent to $B$. The
\textit{block divergence} $E_{B, v}$ measures the influence of an
increment on $v\in\partial B$ when resampling $B$. We will provide a
precise definition of $E_{B, v}$ in Subsection~\ref{ssec:strassen}.

\begin{theorem}\label{theorem:main_result}
  Let $G=(V, E)$ be a finite graph, and $\BB$ be a finite family of
  blocks, such that for every vertex $v\in V$ there is at least one
  $B\in\BB$ with $v\in B$. If there exists $\beta < 1$ such that for
  all $v \in V$
  $$1-\frac{1}{2\lvert\BB\rvert}\left(\#\{ B\in\BB \;\mid\;
    v\in B\}\hspace{0.2cm} - \sum_{B\in\BB\;\mid\;v\in \partial B}
    (E_{B, v}-1)\right)\leq \beta < 1\hspace{0.1cm},$$ then for the mixing
  time $\tau(\varepsilon)$ of the up/down Markov chain $\MM$ on $k$-heights of
  $G$ we have
  $$\tau(\varepsilon) \leq c_{\BB,
    k}\cdot\frac{\left(\left(\log(\frac{1}{\varepsilon})\cdot\lvert
        V\rvert\right)+{\lvert V\rvert}^2
      \cdot\log(k+1)\right)\cdot\log\left(\frac{k\lvert
        V\rvert}{\varepsilon}\right)}{\log(\frac{1}{2\varepsilon})}\hspace{0.1cm},$$
  where 
  $$c_{\BB, k} := \frac{8\cdot
    bmk(k+1)^b}{(1-\beta)\lvert\BB\rvert}\hspace{0.1cm}$$
  with $m := \max_{v\in V} \#\{ B\in\BB \;\mid\; v\in B\}$ and 
  $b := \max_{B\in \BB}\; \lvert B\rvert$.
\end{theorem}

A simplified but weaker version of Theorem \ref{theorem:main_result} is given
by its following corollary:

\begin{corollary}\label{cor:main_result_corollary}
  Let $G=(V, E)$ be a finite graph, and $\BB$ be a finite family of
  blocks such that each vertex $v\in V$ is contained in at least $\check{m}$ blocks
  and in at most $s$ boundaries of blocks, and let
  $E_{\textup{max}} := \max_{B \in \BB, v\in\partial B} E_{B, v}$. If
  there is a~$\beta < 1$ such that
  $$1 - \frac{1}{2\lvert\BB\rvert}\left(\check{m}-s\cdot
    (E_{\textup{max}}-1)\right)\leq\beta\hspace{0.1cm},$$ then the upper bound
  on $\tau(\varepsilon)$ given in Theorem \ref{theorem:main_result} holds.
\end{corollary}

Below, we present various applications of Theorem \ref{theorem:main_result} by
showing that the assumptions are satisfied for different families of graphs
and carefully selected families of blocks $\BB$. Due to limited
computational power, so far we were able to find appropriate upper bounds on
$E_{\text{max}}$ only in the cases $k\in\{2, 3\}$. Note that the up/down
Markov chain $\MM$ is trivially rapidly mixing when $k\in\{0, 1\}$.

A bound for the mixing time of the up/down Markov-chain $\MM$
operating on $k$-heights of $G$ which is of the form
    $$\tau(\varepsilon)<
    c_k\cdot\frac{\left(\left(\log(\frac{1}{\varepsilon})\cdot n\right)+n^2
      \cdot\log(k+1)\right)\cdot\log\left(\frac{kn}{\varepsilon}\right)}
    {\log(\frac{1}{2\varepsilon})}\in\mathcal{O}\left(n^2\log n\right),
    $$
can be shown for the following pairs of an $n$ vertex graph $G$ and a $k$:
\begin{itemize}
\item
  Graph $G$ is a toroidal hexagonal grid graph and $k\in\{2, 3\}$ with
  corresponding constants  $c_2 = 1.747648\cdot 10^5$ and $c_3 = 1.052669\cdot 10^7$.
  (This is Theorem~\ref{theorem:rapidly_mixing_toroidal_hexagonal_graph} in
  Subsection~\ref{ssec:mixingtoroidalhexagonalgraph}.)
\item
  Graph $G$ is a toroidal rectangular grid graph and $k\in\{2, 3\}$ with
  corresponding constants $c_2 = 2.844202\cdot 10^{10}$ and
  $c_3 = 1.333706\cdot 10^{13}$.  (This is
  Theorem~\ref{theorem:rapidly_mixing_toroidal_rectangle_graph} in
  Subsection~\ref{section:rectangular_grid_graphs}.)
\item
  Graph $G$ is a simple $2$-connected $3$-regular planar graph and $k=2$ with
  corresponding constants $c_2 = 4.391132\cdot 10^7$. (This is
  Theorem~\ref{theorem:rapidly_mixing_regular_graph} part (1) in
  Subsection~\ref{subsection:regular_graphs}.)

\item Graph~$G$ is a simple $3$-connected $3$-regular planar graph and $k\in\{2, 3\}$ with corresponding constants $c_2 = 2.195097\cdot 10^7$ and $c_3 = 4.852027 \cdot 10^9$.

If~$G$ is the dual graph of a~$4$-connected triangulation, the constant for the case~$k=2$ can be improved to~\mbox{$c_2 = 1.489256\cdot 10^7$}.

(This is
  Theorem~\ref{theorem:rapidly_mixing_regular_graph} parts (2) and (3) in
  Subsection~\ref{subsection:regular_graphs}.)
\end{itemize}

We conjecture that the up/down Markov-chain is
rapidly mixing on the $k$-heights of these graph classes for all $k$.

Our method can be applied to further families of graphs. The crucial step is
to find an appropriate family of blocks $\BB$ and a sharp analysis of
$E_{B, v}$. Moreover, we believe that \textit{block coupling} can be applied
for proving that several other related Markov chains which also operate on a
distributive lattice structure are rapidly mixing. An example could be
generalized $k$-heights which allow a larger difference between the assigned
values of two adjacent vertices.

Preliminary work of this research can be found in a PhD thesis~\cite{heldt16}.

\subsection{Outline}

In Section~\ref{section:preliminaries} we fix standard terminology related to
Markov chains, give formal definitions for both Markov chains~$\MM$
and~$\MM_{\BB}$ and introduce the proof ingredients for
Theorem~\ref{theorem:main_result}. The proof of
Theorem~\ref{theorem:main_result} itself is presented in
Section~\ref{section:proof_main_result} and consists mainly of defining and
analysing a monotone coupling of the Markov
chain~$\MM_{\BB}$. In Section~\ref{section:applications} we show how to apply
Theorem~\ref{theorem:main_result} and
Corollary~\ref{cor:main_result_corollary} in order to prove
Theorem~\ref{theorem:rapidly_mixing_toroidal_hexagonal_graph},
Theorem~\ref{theorem:rapidly_mixing_toroidal_rectangle_graph}
and~Theorem~\ref{theorem:rapidly_mixing_regular_graph}. Finally, in
Section~\ref{sec:glauber_dynamics} we relate $k$-heights to \textit{spin systems}
and discuss the applicability of a recent result by
Blanca~et~al.~\cite{blancaCaputoZongchenParisiStefankovicVigoda2022}
on bounding the mixing time of corresponding~\textit{Glauber dynamics}.

\section{Preliminaries}\label{section:preliminaries}

\subsection{Markov chains, mixing time and couplings}

For a thorough introduction to the theory of discrete-time Markov chains we
refer the reader to the recent book by Levin, Peres \& Wilmer \cite{LPW17} or
to the book by Jerrum \cite[ch. 3-5]{jerrum03}. Throughout this article, all
Markov chains $\mathcal{C} = (X_t)_{t\in\mathbb{N}}$ are discrete-time Markov
chains on some finite state space $\mathcal{X}$. Moreover, they fulfill the
following properties:
\begin{itemize}
\item \textit{Homogeneous}: The transition probabilities
  $$\mathcal{C}(x, y) := \Prob[X_{t+1} = y\;\mid\; X_t = x]$$ are
  constant over time, i.e.~independent of $t$.
\item \textit{Irreducible}: Each state can be reached from any other state in
  a sequence of transitions with non-zero probability.
\item \textit{Aperiodic}: The possible number of steps for returning from a
  state to the same state again are not only multiples of some period
  $m\in\mathbb{N}, m\geq 2$.
\item \textit{Symmetric}: We have $\mathcal{C}(x, y) = \mathcal{C}(y, x)$ for
  all $x, y\in \mathcal{X}$.
\end{itemize}
It is well-known that under these conditions the distribution of $X_t$
converges towards the uniform distribution on $\mathcal{X}$, which we denote
by $\mathcal{U}(\mathcal{X})$. We want to analyse whether this convergence
happens in polynomial time. To make this precise, let
$$\mathcal{C}^t(x, y) := \Prob[X_t = y \;\mid\; X_0 = x]$$ and let
$\mathcal{C}^t(x, \cdot)$ be the distribution of $X_t$ when $\mathcal{C}$
starts in state $X_0 = x$. For any two probability measures $\mu, \mu'$ on
$\mathcal{X}$ the \textit{total variation distance} is defined as
$$\lVert \mu - \mu'\rVert_{TV} := \max_{A\subset \mathcal{X}} \;\lvert \mu(A)
- \mu'(A) \rvert\hspace{0.1cm}.$$ Now,
$$d(t) := \max_{x\in\mathcal{X}} \;\lVert \mathcal{C}^t(x, \cdot) -
\mathcal{U}(\mathcal{X}) \rVert_{TV}$$ measures the worst-case distance of
$\mathcal{C}$ to the uniform distribution after $t$ steps. Hence, the
\textit{mixing time} defined as
$$\tau(\varepsilon) := \min \{t\in\mathbb{N}\;\mid\; d(t) < \varepsilon
\}\hspace{0.1cm},$$ measures the number of time steps required for
$\mathcal{C}$ to be $\varepsilon$-close to the uniform distribution. We say
that $\mathcal{C}$ is \textit{rapidly mixing}, if
$\tau_{\mathcal{C}}(\varepsilon)$ is upper bounded by a polynomial
in~$\log(\lvert\mathcal{X}\rvert)$ and~$\log(\varepsilon^{-1})$.

A common technique to prove that a Markov chain is rapidly mixing is by using
a coupling. A \textit{coupling} of a Markov chain $\mathcal{C}$ is another
Markov chain~$(X_t, Y_t)$ operating on $\mathcal{X}\times\mathcal{X}$ whose
components $(X_t)$ and $(Y_t)$ are copies of~$\mathcal{C}$, i.e., their
transition probabilities are the same as those of~$\mathcal{C}$. The Markov
chains~$(X_t)$ and~$(Y_t)$ are typically not independent though. If there is a
partial order $\leq$ defined on $\mathcal{X}$, then we call the coupling a
\textit{monotone coupling}, if
$$\Prob[X_{t+1} \leq Y_{t+1} \;\mid\; X_t \leq Y_t] = 1.$$ If
$X_0 \leq Y_0$, this implies $X_t \leq Y_t$ for all $t$. We use the term
(monotone) coupling also to refer to the random transition
$(X_t, Y_t)\mapsto (X_{t+1}, Y_{t+1})$ that defines a (monotone) coupling
$(X_t, Y_t)_{t\in\mathbb{N}}$.

One usually aims for a monotone coupling in which in each transition~$X_t$
and~$Y_t$ get closer in expectation with respect to some distance measure,
since then an upper bound on the mixing time of $\mathcal{C}$ is provided by
the classical result in Theorem \ref{theorem:montone_coupling_classical}.

\begin{theorem}[Dyer \& Greenhill, Theorem 2.1 in \cite{dyerGreenhill}]
  \label{theorem:montone_coupling_classical}
  Let $(X_t, Y_t)\mapsto (X_{t+1}, Y_{t+1})$ be a coupling of a Markov chain $\mathcal{C}$ operating on
  a state space $\mathcal{X}$, let
  $d:\mathcal{X}\times\mathcal{X}\to\mathbb{N}_0$ be any integer value metric
  and let
  $$D := \max \{ d(x, y) \;\mid\; (x, y)\in\mathcal{X}\times\mathcal{X}\}.$$
  Suppose there exists $\beta < 1$ such that
  $$\mathbb{E}[d(X_{t+1}, Y_{t+1})] < \beta\cdot d(X_t, Y_t)$$ for all
  $t$. Then the mixing time $\tau_\mathcal{C}(\varepsilon)$ is upper bounded
  by
  $$\tau_\mathcal{C}(\varepsilon) \leq
  \frac{\log\left(\frac{D}{\varepsilon}\right)}{1-\beta}\hspace{0.1cm}.$$
\end{theorem}

For the Markov chain $\MM$ on $k$-heights discussed in
the introduction we have no monotone coupling for which
Theorem~\ref{theorem:montone_coupling_classical} can directly be applied. This is the
reason for introducing a boosted Markov chain~$\MM_{\BB}$
on~$k$-heights, in which a transition is typically changing the values on
a larger set of vertices. For finding a monotone coupling, we use the
\textit{path coupling} technique introduced by Bubley and Dyer
\cite{bubleyDyer1997}. That is we define a monotone coupling
\mbox{$(X_t, Y_t)\to (X_{t+1}, Y_{t+1})$} on pairs of $k$-heights that differ by one on
a single vertex. Using the following theorem this coupling
can be extended to arbitrary pairs of~$k$-heights.

\begin{theorem}[Dyer \& Greenhill, Theorem 2.2 in \cite{dyerGreenhill}]
    \label{theorem:path_coupling}
    Suppose $\mathcal{C}$ is a Markov chain operating on $\mathcal{X}$. Let
    $\dd$ be an integer valued metric defined on $\mathcal{X}\times\mathcal{X}$
    which takes values in $\{0, \ldots, D\}$, and let $S$ be a subset of
    $\mathcal{X}\times\mathcal{X}$ such that for all
    $(x, y)\in\mathcal{X}\times\mathcal{X}$ there exists a path
    $$\gamma_{x, y}: x = x_0, x_1, \ldots, x_r= y$$ with
    $(x_i, x_{i+1})\in S$, that is a shortest path, i.e.,
    $$\sum_{l=0}^{r-1} \dd(x_l, x_{l+1}) = \dd(x, y).$$ Let
    $(x, y)\mapsto (x', y')$ be a coupling of $\mathcal{C}$ that is defined
    for all $(x, y)\in S$. For any~\mbox{$(x, y)\in\mathcal{X}\times\mathcal{X}$},
    apply this coupling along the path $\gamma_{x, y}$ to obtain a new path
    $x' = x_0', \ldots, x_r' = y'$. Then, \mbox{$(x, y) \mapsto (x', y')$} defines a
    coupling of $\mathcal{C}$ on all tuples
    $(x, y)\in\mathcal{X}\times\mathcal{X}$. Moreover, if there exists~\mbox{$\beta < 1$} so that $$\mathbb{E}[\dd(x', y')] < \beta\cdot \dd(x, y)$$ for all
    $(x, y)\in S$, then the same inequality holds for all
    \mbox{$(x, y)\in\mathcal{X}\times\mathcal{X}$} in the extended coupling.
\end{theorem}

For some families of graphs and choices of blocks this theorem allows to
conclude that $\MM_{\BB}$ is rapidly mixing. A transition
of $\MM_{\BB}$ can be simulated by a sequence of transitions of
the original Markov chain $\MM$. Using the \textit{comparison
  technique} that is manifested in the following theorem we can push a
mixing result from~$\MM_{\BB}$ to the original chain 
$\MM$.

\begin{theorem}[Randall \& Tetali, Theorem 3 in \cite{randallTetali00}]
    \label{theorem:comparison_technique}
    Let $\mathcal{C}$ and $\widetilde{\mathcal{C}}$ be two reversible Markov
    chains on the same state space $\mathcal{X}$ and having the same
    stationary distribution $\pi$. Let $E(\mathcal{C})$ be the set of transitions
    of~$\mathcal{C}$ and $E(\widetilde{\mathcal{C}})$ be the set of
    transitions of $\widetilde{\mathcal{C}}$.

    Suppose that for each transition $(x, y)\in E(\widetilde{\mathcal{C}})$
    there is a path $\gamma_{x, y}: x=x_0, \ldots, x_k = y$ of
    transitions $(x_i, x_{i+1})\in E(\mathcal{C})$. For a transition
    $(u, v)\in E(\mathcal{C})$ let
    $$\Gamma(u, v) := \left\{ (x, y)\in E(\widetilde{\mathcal{C}}) \;\mid\;
      (u,v) \in \gamma_{x, y} \right\}\hspace{0.1cm},$$ and let
    $$A := \max_{(u, v)\in E(\mathcal{C})}\left\{ \frac{1}{\pi(u)\mathcal{C}(u, v)}
    \sum_{(x, y)\in\Gamma(u, v)}\lvert\gamma_{x,y}\rvert\hspace{0.1cm}\pi(x)\;\widetilde{\mathcal{C}}(x,y)\hspace{0.1cm}\right\},
    $$
    where $\lvert\gamma_{x, y}\rvert$ denotes the length of $\gamma_{x, y}$
    and
    $\mathcal{C}(u, v) := \Prob_{\mathcal{C}}[X_{t+1} = v\mid X_t = u]$
    is the probability of the transition $(u, v)$ in $\mathcal{C}$. Then the
    mixing time $\tau_{\mathcal{C}}$ of $\mathcal{C}$ can be bounded in terms
    of the mixing time $\tau_{\widetilde{\mathcal{C}}}$ as by
    $$\tau_{\mathcal{C}}(\varepsilon) \leq
    \frac{4\log\left({1}/{(\varepsilon\cdot\pi_{\text{min}})}\right)}
           {\log\left({1}/{(2\varepsilon)}\right)}\cdot
           A\cdot\tau_{\widetilde{\mathcal{C}}}(\varepsilon),
    $$
    where $\pi_{\text{min}} := \min \{ \pi(x)\;\mid\; x\in\mathcal{X}\}$.
\end{theorem}

\subsection{The up/down Markov chain}

For a graph $G=(V, E)$ let $\Omega_G^k$ denote the set of $k$-heights of
$G$. The \textit{up/down Markov chain} $\MM(G, k)$ operates
on $\Omega_G^k$. It starts with any $k$-height~$X_0\in\Omega_G^k$ and its
transitions are given by Algorithm \ref{alg:principal_chain_transition}. We
usually consider $G$ and $k$ as fixed parameters, which is why we write just
$\Omega$ and $\MM$ for simplicity.

\begin{algorithm}
\DontPrintSemicolon
\caption{Transition of up/down Markov chain $\MM$: $X_t \to X_{t+1}$}
\label{alg:principal_chain_transition}
Sample $v\in V$, and $\triangle\in\{-1, 1\}$, and $p\in [0, 1]$ uniformly at random
\medskip

$\varphi(v) := \begin{cases}
X_t(v)+\triangle & v = \Tilde{v}\\
X_t(v) & v \neq \Tilde{v}
\end{cases}$\;
\smallskip
\eIf{$\varphi$ \textsf{\upshape is a valid $k$-height} \textbf{\upshape and} $p \leq \frac{1}{2}$}
{
    $X_{t+1}\gets \varphi$\;
}{
    $X_{t+1}\gets X_t$\;
  }
\Ret $X_{t+1}$
\smallskip
\end{algorithm}

Performing a transition only if $p\leq \frac{1}{2}$ is known as making the chain \emph{lazy}, it ensures aperiodicity of the Markov chain.

For two $k$-heights $X, Y\in\Omega$ we write $X\leq Y$ if $X(v)\leq Y(v)$ for
all $v\in V$. This makes $\Omega$ a poset. Furthermore, equipped with the
operations
$$
 (X\land Y)(v) := \min \{X(v), Y(v)\} \textrm{ and } (X\lor Y)(v) := \max \{X(v), Y(v)\}
$$
the set $\Omega$ becomes a distributive lattice. The chain $\MM$ can be seen as
a random walk on the diagram of $\Omega$. For $X, Y\in \Omega$ we introduce
the distance between $X$ and $Y$ defined as
$$\dd(X, Y) = \sum_{v\in V} \lvert X(v) - Y(v)\rvert .$$

\begin{lemma}\label{lemma:path_lengths}
  Let $X, Y\in\Omega$. Then $\dd(X, Y)$ is the smallest number of transitions
  of $\MM$ to get from state $X$ to state $Y$.
\end{lemma}
\begin{proof}
  In each step of the Markov chain $\MM=(X_t)$, the value $d(X_t, Y)$
  changes by at most one. Hence, $\MM$ cannot reach $Y$ from $X$ in
  less than $\dd(X, Y)$ steps.
    
  Suppose $X\neq Y$. We claim that there is a $X'\in\Omega$, such that $(X, X')$
  is a transition of~$\MM$ and~$\dd(X', Y) < \dd(X, Y)$.
  By symmetry, we can assume that $S:=\{ v \mid X(v) < Y(v)\}$ is nonempty and
  choose $v_0\in S$ with $X(v_0)$ being minimal. 
  If increasing the value of $X$ at $v_0$ results in
  a valid $k$-height we have found $X'$. Otherwise there must be a vertex $v_1$ adjacent
  to~$v_0$ with~$X(v_1)<X(v_0)$. Because of this and because $Y$ is a $k$-height, 
  \[
  X(v_1) \leq X(v_0) - 1 < Y(v_0) - 1 \leq Y(v_1)\;.
  \] 
  Hence, $v_1\in S$ and $X(v_1) < X(v_0)$. This contradicts the choice of $v_0$.
\end{proof}

For the Markov chain $\MM$ a natural monotone coupling is given by
using the same random vertex $v$, random $p$, and offset $\triangle$ for both
$X_t$ and $Y_t$. Unfortunately, this coupling does not satisfy
$\mathbb{E}[\dd(X_{t+1}, Y_{t+1})] < \dd(X_t, Y_t)$. For an example consider a
2-path $a$ --- $b$ --- $c$ as graph and $k=3$, if
$X_t(a)=Y_t(a)=X_t(c)=Y_t(c)=1$, $X_t(b)=0$, and $Y_t(b)=2$, then
$\dd(X_t, Y_t) =2$ and
$\mathbb{E}[\dd(X_{t+1}, Y_{t+1})] = \frac{1}{2} 2 + \frac{1}{2}( \frac{4}{6}3
+ \frac{2}{6}1) = \frac{13}{6} > 2$.

\subsection{The block Markov chain}\label{subsection:preliminaries_blockchain}

In this section $\BB$ will always be a family of \textit{blocks} of a
graph $G$, that is $\BB$ is a (multi)set of subsets of the vertices of
$G$ which forms a cover, i.e., for each vertex $v\in V$ there is a
$B\in\BB$ with $v\in B$.

Typically a family of blocks consists of well connected subsets of the graph.
For instance, if~$G$ is a $a\times b$ grid, then
$\BB$ could be the family of all $4\times 4$ subgrids.

The \textit{boundary} of a block $B\in\BB$ is the set $\partial B :=
\left\{ v\in V\setminus B\;\mid\; \{v, w\}\in E\text{ for some }w\in B\right\}$.
With $\Omega_B$ we denote the set of $k$-heights of the subgraph of $G$ induced
by $B$, i.e.,
$$
  \Omega_B := \left\{\varphi: B\to \{0, \ldots, k\} \;\mid\;\varphi\textit{ }k\text{-height w.r.t. }G[B]\right\}.
$$ 

For~$X\in\Omega$ and any~$\varphi: B\to \{0, \ldots, k\}$, we
define~$\REP X + \varphi!: V\to\{0, \ldots, k\}$ as the assignment which maps
a $v\in B$ to $\varphi(v)$ and a $v\in V\setminus B$ to $X(v)$.

With $\Omega_{B\vert X}$ we denote the set of \textit{admissible fillings} of
$B$ in $X$, this set consists of all $k$-heights in $\Omega_B$ which extend $X$, i.e., 
$$
\Omega_{B\vert X} := \left\{ \varphi\in\Omega_B\;\mid\; \REP X+\varphi!\in\Omega \right\}.
$$
Note that if two~$k$-heights~$X, X'\in\Omega_B$ agree on~$\partial B$,
i.e.,~$X(v)=X'(v)$ for all~$v\in\partial B$, then,
\mbox{$\Omega_{B\vert X} = \Omega_{B\vert X'}$}. This allows us to
use the notation $\Omega_{B\vert X}$ also in the case where the $k$-height $X$ is only
defined on $\partial B$, we call such a $X\in\Omega_{\partial B}$ a
\textit{boundary constraint}. A boundary constraint~$X\in\Omega_{\partial B}$
is called \textit{extensible}, if~$\Omega_{B\vert X}\neq\emptyset$.

For a fixed family of blocks $\BB$ the \textit{block Markov chain}
$\MM_{\BB}$ operates on the set of $k$-heights $\Omega$.
A transition depends on a
block $B\in\BB$ and an admissible filling
$\varphi\in\Omega_{B\vert X_t}$ chosen at random:

\begin{algorithm}
\DontPrintSemicolon
\caption{Transition of block Markov chain $\MM_{\BB}$: $X_t \to X_{t+1}$}
\label{alg:block_chain_transition}
Sample $B\in \BB$, and $\varphi\in\Omega_{B\vert X_t}$, and $p\in [0, 1]$ uniformly at random\;
\medskip

\uIf{$p \leq \frac{1}{2}$}
{
    $X_{t+1} \gets \REP X_t + \varphi!$\;
}
\Else{
    $X_{t+1} \gets X_t$\;
}
\Ret $X_{t+1}$
\smallskip
\end{algorithm}

Chain~$\MM_{\BB}$ can be seen as a boosted version of the
up/down chain~$\MM$, as in each transition the values of an entire
block are updated. Assuming that all blocks of~$\BB$ are of constant
small size, one can implement a computer simulation
of~$\MM_{\BB}$ efficiently by preprocessing all
sets~$\Omega_{B\lvert X}$. Hence, the mixing behaviour
of~$\MM_{\BB}$ is a problem of independent interest. For us, however,
$\MM_{\BB}$ will mainly serve as an auxiliary tool for
proving that the up/down chain $\MM$ is rapidly mixing.

The existence of a monotone coupling of $\MM_{\BB}$ is
non-trivial, it will be the main part of our proof in Section
\ref{section:proof_main_result}. The fact that monotone couplings exist for
$\MM$ as well as for $\MM_{\BB}$ enables us to
directly apply \textit{coupling from the past} introduced by Propp and
Wilson~(see \cite{LPW17}) for uniform sampling from $\Omega$ and for empirically
estimating the mixing times $\tau_{\MM}(\varepsilon)$ and
$\tau_{\MM_{\BB}}(\varepsilon)$.

\subsection{A discrete version of a theorem by Strassen
                   and the block divergence}
\label{ssec:strassen}

Let $P$ be a finite partially ordered set and $\mu_1, \mu_2: P\to [0, 1]$
probability distributions on the elements of $P$. A set $U\subset P$ is an
\textit{upset of $P$} if $x\in U$ and~$x\leq y$ implies
$y\in U$. We say $\mu_1$ is \textit{stochastically dominated} by~$\mu_2$, if
$\mu_1(U)\leq \mu_2(U)$ for all upsets $U\subset P$.

The following theorem is a discrete application of a theorem by Strassen
\cite[Theorem 11]{strassen65}. Its application to Markov chains is also
covered in a very accessible way in \cite{lindvall99}. In the problem set
\cite{lalley07} a purely discrete proof using a MinFlow-MaxCut argument is
suggested.

\begin{theorem}\label{theorem:discrete_strassen}
  Let $\mu_1$ and $\mu_2$ be probability distributions on a finite partially
  ordered set $P$ such that $\mu_1$ is stochastically dominated by
  $\mu_2$. Then there exists a probability distribution $\lambda$ on
  $P\times P$ with the following properties:
    \begin{itemize}
    \item[1)] $\lambda$ is a joint distribution of $\mu_1$ and $\mu_2$,
      i.e., \begin{align*}
  \forall x\in P: \sum_{y\in P} \lambda(x, y) &= \mu_1(x)\text{,\hspace{0.1cm} and}\\
  \forall y\in P: \sum_{x\in P} \lambda(x, y) &= \mu_2(y)\hspace{0.5mm}.
        \end{align*}
        \item[2)] If $\lambda(x, y) > 0$, then $x \leq y$ in $P$.
    \end{itemize}
\end{theorem}
We will make use of Theorem \ref{theorem:discrete_strassen} later when
constructing a monotone coupling~$(X_t, Y_t)$ of~$\MM_{\BB}$.
For the rapid convergence between $X_t$ and $Y_t$, the \textit{block
  divergence}, as introduced in the following, plays a key role.

We call a pair $(X, Y)\in\Omega\times\Omega$ of $k$-heights a \textit{cover
pair}, if $X\leq Y$ and~$\dd(X, Y)=1$, i.e., $X$ and~$Y$ differ at a
single vertex $v$ where \mbox{$Y(v) = X(v) + 1$}. Let $B\in\BB$ be
some block. If $v\in \partial B$, then the sets of admissible
fillings~$\Omega_{X\vert B}$ and~$\Omega_{Y\vert B}$ may differ. Let
$\mathcal{U}(\Omega_{B\vert X})$ and $\mathcal{U}(\Omega_{B \vert Y})$ denote
the uniform distributions on~$\Omega_{B\vert X}$ and~$\Omega_{B\vert Y}$,
respectively. We can view $\mathcal{U}(\Omega_{B\vert X})$
and~$\mathcal{U}(\Omega_{B\vert Y})$ as distributions on
the partially ordered set $\Omega_B$. Later we will see
that~$\mathcal{U}(\Omega_{B\vert X})$ is stochastically dominated by
$\mathcal{U}(\Omega_{B\vert Y})$. Hence,
Theorem~\ref{theorem:discrete_strassen} provides a distribution $\lambda_{B, X, Y}$
on $\Omega_B \times \Omega_B$, which in fact is a distribution
on~$\Omega_{B\vert X}\times \Omega_{B\vert Y}$, as other elements of
$\Omega_B \times \Omega_B$ have zero probability.

When constructing the monotone coupling of $\MM_{\BB}$ we aim
for a rapid convergence of $X_t$ and~$Y_t$. For this it will turn out to be
crucial that when~$(X', Y')$ is drawn from $\lambda_{B, X, Y}$, then the
distance~\mbox{$\dd(X', Y') := \sum_{v\in B} Y'(v) - X'(v)$} is small in
expectation. We call this quantity
$$
  E_{B, v} := \max\left\{ \mathbb{E}_{\lambda_{B, X, Y}}[\dd(X', Y')]
   \;\middle|\;\begin{aligned}
                &(X, Y)\in\Omega\times\Omega\text{ cover pair},\\
                &Y(v) = X(v) + 1
               \end{aligned}\right\}
$$
the \textit{block divergence} for block $B\in\BB$ and boundary vertex
$v\in\partial B$.

The distribution $\lambda_{B, X, Y}$ only depends on the values that $X$ and
$Y$ take on $\partial B$; so for computing $E_{B, v}$, we only need to
maximize over
pairs~\mbox{$(X, Y)\in \Omega_{\partial B}\times\Omega_{\partial B}$} of
extensible boundary constraints that are cover relations on $\partial B$. But
how can we compute $\mathbb{E}_{\lambda_{B, X, Y}}[d(X', Y')]$?
Lemma~\ref{lemma:linearity_expectation} gives an answer.

For an admissible filling $\varphi\in\Omega_B$, let
$w(\varphi) := \sum_{v\in V} \varphi(v)$ be its \textit{weight}.

\begin{lemma}\label{lemma:linearity_expectation}
  Let $(X, Y)\in\Omega\times\Omega$ be a cover relation and $B\in\BB$
  some block. Let~$\varphi_1 \sim \mathcal{U}(\Omega_{B\vert X})$ and
  $\varphi_2 \sim \mathcal{U}(\Omega_{B\vert Y})$ be random variables drawn uniformly
  from the admissible fillings of $B$ with respect to $X$ and $Y$, and $k$. Then it
  holds:
  $$\mathbb{E}_{\lambda_{B, X, Y}}[\dd(X', Y')] = \mathbb{E}[w(\varphi_2)] -
  \mathbb{E}[w(\varphi_1)]\hspace{0.1cm}$$
\end{lemma}
\begin{proof}
  In the calculation we use that $(X',Y')$ is a pair drawn from the
  distribution $\lambda_{B, X, Y}$ with marginals $\Omega_{B\vert X}$ and
  $\Omega_{B\vert Y}$ provided by Theorem \ref{theorem:discrete_strassen}. In
  particular $X' \leq Y'$, i.e., $X'(v) \leq Y'(v)$ for all $v$.
\begin{align*}
  \mathbb{E}_{\lambda_{B, X, Y}}[d(X', Y')] &=
      \sum_{v\in B} \mathbb{E}_{\lambda_{B, X, Y}}[Y'(v)] -
                   \mathbb{E}_{\lambda_{B, X, Y}}[X'(v)] \\
   &= \sum_{v\in B} \mathbb{E}_{\mathcal{U}(\Omega_{B\vert Y})}[Y'(v)] -
                   \mathbb{E}_{\mathcal{U}(\Omega_{B\vert X})}[X'(v)] \\
   &= \mathbb{E}[w(\varphi_2)] - \mathbb{E}[w(\varphi_1)].
\end{align*}
\vskip-1.3em
\end{proof}

\section{Mixing time of up/down and block Markov chain}
\label{section:proof_main_result}

\subsection{Stochastic dominance in block Markov chain}

To apply Theorem~\ref{theorem:discrete_strassen} in our context we
need to verify stochastic dominance between $\mathcal{U}(\Omega_{B\vert X})$ 
and $\mathcal{U}(\Omega_{B\vert Y})$. This will be done
in Proposition~\ref{prop:stoch-dom}. In the proof we make use of the
Ahlswede--Daykin \textit{4~Functions Theorem}.

\begin{lemma}[4 Functions Theorem]
  \label{lemma:four_functions}
  Let $D$ be a distributive lattice and
  $f_1, f_2, f_3, f_4: D\to\mathbb{R}_{\geq 0}$,
  such that for all $a, b\in D$:
$$f_1(a) f_2(b) \leq f_3(a \lor b) f_4(a\land b).$$
Then for all $A, B\subset D$:
$$f_1(A)f_2(B)\leq f_3(A\lor B)f_4(A\land B),$$
where $f_i(A) = \sum_{a\in A}f_i(a)$,
$A\lor B = \{ a\lor b\;\mid\; a\in A, b\in B\}$
and~$A\land B = \{ a\land b\;\mid\; a\in A, b\in B\}$.
\end{lemma}

The original proof of the 4 Functions Theorem can be found in
Ahlswede and Daykin~\cite{ahlswedeDaykin78},
another source is \textit{The Probabilistic Method} by Alon and 
Spencer~\cite{alonSpencer16}.

\begin{lemma}\label{lemma:downsets_and_upsets}
  Let $X, Y\in\Omega$, $X\leq Y$ be $k$-heights of $G=(V, E)$, and let
  $B\subset V$ be a block. Let~$D$ be the smallest distributive sublattice of
  $\Omega_B$ containing $\Omega_{B\vert X}\cup\Omega_{B\vert
    Y}$. Then~$\Omega_{B\vert X}$ forms a downset and~$\Omega_{B\vert Y}$
  forms an upset in~$D$.
\end{lemma}
\begin{proof}
  By symmetry, it suffices to show that $\Omega_{B\vert X}$ is a downset in
  $D$. So let~$g, h\in D$, with~$g\leq h$ and $h\in\Omega_{B\vert X}$. We have to
  show $g\in\Omega_{B\vert X}$.

  Suppose $g\notin\Omega_{B\vert X}$, then since $g\in \Omega_B$, we must
  have~\mbox{$\lvert g(v) - X(v') \rvert > 1$} for two adjacent vertices
  $v\in B$ and \mbox{$v'\in\partial B$}.  With~$h\in\Omega_{B\vert X}$ we
  conclude~\mbox{$g(v) \leq h(v) \leq X(v') + 1$} so that
    $$g(v) < X(v') - 1\hspace{0.05cm}.$$

    For all $f\in\Omega_{B\vert X}$ we have $f(v) \geq X(v') - 1$ by definition.
    Also for all $f\in\Omega_{B\vert Y}$ we have~$f(v) \geq Y(v') - 1$ and using
    $Y \geq X$ we get $f(v) \geq X(v') - 1$ as well. Therefore,
    $$(\min D)(v) = \min \{ f(v) \;\mid\; f\in \Omega_{B\vert X}\cup\Omega_{B\vert Y} \}
    \geq X(v') - 1 > g(v).$$
    This is a contradiction to $g\in D$.
\end{proof}

\begin{proposition}\label{prop:stoch-dom}
  Let $X, Y\in\Omega$, with $X\leq Y$ be $k$-heights of $G=(V, E)$, and let
  $B\subset V$ be a block. Then on $\Omega_B$,
  $\mathcal{U}(\Omega_{B\vert X})$ is stochastically dominated by
  $\mathcal{U}(\Omega_{B\vert Y})$.
\end{proposition}
\begin{proof}
  Let $D$ be the smallest distributive sublattice of $\Omega_B$ containing
  $\Omega_{B\vert X}\cup\Omega_{B\vert Y}$, and consider
  $\mathcal{U}(\Omega_{B\vert X})$ and $\mathcal{U}(\Omega_{B\vert Y})$ as
  distributions on $D$ with zero probability outside of $\Omega_{B\vert X}$
  and~$\Omega_{B\vert Y}$, respectively. In particular,~$\Omega_{B\vert X}$
  and~$\Omega_{B\vert Y}$ have zero probability on $\Omega_B\setminus
  D$. Therefore, it suffices to show
  that~$\Omega_{B\vert X}(U) \leq \Omega_{B\vert Y}(U)$ for any upset~$U$
  in~$D$. This is equivalent~to
  $$0 \leq \mathcal{U}(\Omega_{B\vert Y})(U) - \mathcal{U}(\Omega_{B\vert
    X})(U) = \frac{\lvert U\cap \Omega_{B\vert Y}\rvert}{\lvert\Omega_{B\vert
      Y}\rvert}-\frac{\lvert U\cap \Omega_{B\vert
      X}\rvert}{\lvert\Omega_{B\vert X}\rvert}\hspace{0.05cm}.$$ Define 
  four functions $f_1, f_2, f_3, f_4: D\to\mathbb{R}_{\geq 0}$ as
  \vskip-8mm\vbox{}
  
  \begin{center}
  \parbox{.3\linewidth}{
    \begin{align*}
      f_1(h) &:= \chi_{U\cap\Omega_{B\vert X}}(h)\\
      f_3(h) &:= \chi_{U\cap\Omega_{B\vert Y}}(h)
    \end{align*}}           
  \parbox{.3\linewidth}{
    \begin{align*}
       f_2(h) &:= \chi_{\Omega_{B\vert Y}}(h)\\
       f_4(h) &:= \chi_{\Omega_{B\vert X}}(h)
    \end{align*}}
  \end{center}
  for all $h\in D$, where $\chi_S$ denotes the characteristic function of the
  set $S$. By Lemma~\ref{lemma:downsets_and_upsets}, $\Omega_{B\vert X}$ forms
  a downset and $\Omega_{B\vert Y}$ forms an upset in $D$. Aiming for an
  application of Lemma~\ref{lemma:four_functions}, we have to verify
  that $$f_1(h) f_2(g) \leq f_3(h \lor g) f_4(h\land g)$$ for all $h, g\in
  D$. If $f_1(h) f_2(g) = 0$, this holds trivially, so we assume
  \mbox{$f_1(h) = f_2(g) = 1$}. This implies $h\in U\cap \Omega_{B\vert X}$
  and~\mbox{$g\in\Omega_{B\vert Y}$}. Because~\mbox{$h\lor g\geq g$}
  and~$\Omega_{B\vert Y}$ is an upset,~\mbox{$h\lor g \in \Omega_{B\vert Y}$}, and
  because~$h\lor g\geq h$ and~$U$ is an upset,~\mbox{$h\lor g\in U$}. Hence,
  $h\lor g\in U\cap\Omega_{B\vert Y}$ and~\mbox{$f_3(h\lor g) = 1$}. Moreover,
  $h\land g \leq h\in\Omega_{B\vert X}$
  implies~$h\land g\in \Omega_{B\vert X}$, because~$\Omega_{B\vert X}$ is a
  downset. Hence, we also get~\mbox{$f_4(h\land g) = 1$}. Therefore,
  $f_3(h\lor g)f_4(h\land g) = 1$ and the assumption of
  Lemma~\ref{lemma:four_functions} holds. Lemma~\ref{lemma:four_functions}
  applied on $A=B=D$ yields

\begin{align*}
  0 &\leq f_3(D)f_4(D) - f_1(D)f_2(D)
      = \lvert U\cap\Omega_{B\vert Y}\rvert\cdot
    \lvert\Omega_{B\vert X}\rvert-\lvert U\cap\Omega_{B\vert X}\rvert\cdot
    \lvert\Omega_{B\vert Y}\rvert,
\end{align*}
division by $\lvert\Omega_{B\vert X}\rvert\cdot\lvert\Omega_{B\vert Y}\rvert$
yields the inequality we need.
\end{proof}

\subsection{Block coupling}

We are now ready to define the block coupling, which is a monotone coupling
$(X_t, Y_t)$ of the block Markov chain $\MM_{\BB}$.
Recall that $X_t \leq Y_t$ is a cover relation if $d(X_t, Y_t)=1$. In
this situation, we will randomly select a block $B\in\BB$ and
compute the joint distribution of~$\mathcal{U}(\Omega_{B\vert X_t})$ and $\mathcal{U}(\Omega_{B\vert Y_t})$ as described
in Theorem \ref{theorem:discrete_strassen}. A sample from this joint distribution
yields $k$-heights $X_{t+1}$, $Y_{t+1}$ which also satisfy $X_{t+1} \leq Y_{t+1}$. Hence, 
we have a monotone coupling $(X_t, Y_t)\mapsto (X_{t+1}, Y_{t+1})$ on the set
$$
S := \{ (X, Y)\in\Omega\times\Omega\;\mid\; (X, Y)\text{ is a cover relation}\}
$$
By Theorem \ref{theorem:path_coupling} the coupling defined on
$S\times S$ can be extended to a coupling defined on $\Omega\times\Omega$.
From the construction of the extended coupling it follows that the
monotonicity of the coupling on $S\times S$ is inherited by the extended
coupling. The transition  $(X_t, Y_t)\mapsto (X_{t+1}, Y_{t+1})$ is detailed
in Algorithm~\ref{alg:block_coupling}.

\begin{algorithm}[htb]
  \DontPrintSemicolon
\caption{Monotone coupling $(X_t \leq Y_t)\mapsto (X_{t+1} \leq Y_{t+1})$ of $\MM_{\BB}$}
\label{alg:block_coupling}
Sample $p\in [0, 1]$ uniform at random\;
\medskip

\uIf{$p \leq \frac{1}{2}$}
  {
    \uIf{$d(X_t, Y_t) \leq 1$}
    {
    Sample $B\in \BB$ uniformly at random\;
    \uIf{$X_t(v) = Y_t(v)$ \textup{for all} $v\in\partial B$}
    {
      Sample $\varphi$ from $\mathcal{U}(\Omega_{B\vert X_t})$ \Comment*[r]{Note that $\mathcal{U}(\Omega_{B\vert X_t})=\mathcal{U}(\Omega_{B\vert Y_t})$}
      $(X_{t+1}, Y_{t+1}) \gets (\REP X_t + \varphi!,\REP Y_t + \varphi!)$\;
    }
    \Else{
      Apply Strassen's Theorem on $\mathcal{U}(\Omega_{B\vert X_t})$
      and~$\mathcal{U}(\Omega_{B\vert Y_t})$\\
      Sample $(\varphi_X, \varphi_Y)$ from the
      joint distribution $\lambda_{B, X_t, Y_t}$ on
      $\Omega_{B\vert X_t}\times\Omega_{B\vert Y_t}$\;
        $(X_{t+1}, Y_{t+1}) \gets (\REP X_t + \varphi_X!,\REP Y_t + \varphi_Y!)$\;
    }
   }
   \Else{
     Define $(X_{t+1}, Y_{t+1})$ using path coupling (Theorem~\ref{theorem:path_coupling})\;}
   }
\Else{
  $(X_{t+1}, Y_{t+1}) \gets (X_t, Y_t)$\;}
\Ret $(X_{t+1}, Y_{t+1})$
\smallskip
\end{algorithm}

The intermediate step of first defining a coupling on cover relations before
extending it via path coupling is for the sake of the analysis of the mixing
time. In fact, as we have stochastic dominance
between~$\mathcal{U}(\Omega_{B\vert X_t})$
and~$\mathcal{U}(\Omega_{B\vert Y_t})$ for any~$X_t\leq Y_t$, not just for
cover relations, so we could use Strassen's Theorem
(Theorem~\ref{theorem:discrete_strassen}) directly to define a monotone
coupling on $\Omega\times\Omega$. However, it is easier to
analyse~$\mathbb{E}[d(X_{t+1}, Y_{t+1})]$ when~$X_t\leq Y_t$ is a cover
relation and to rely on path coupling (Theorem~\ref{theorem:path_coupling})
for the general case.

\begin{lemma}\label{lemma:one_step_block_coupling_distance}
  If $\beta$ is chosen such that
  $$
  1-\frac{1}{2\lvert\BB\rvert}\left(\#\{ B\in\BB
    \mid v\in B\}\hspace{0.2cm} -
    \sum_{B\in\BB\hspace{0.1cm}\mid\hspace{0.1cm}v\in \partial B}
    (E_{B, v}-1)\right)\leq \beta\hspace{0.1cm}
  $$
  for all vertices $v \in V$, and $X\leq Y$ is a cover relation in
  $\Omega\times\Omega$, and $(X,Y)\mapsto (X', Y')$ is a transition of the
  block coupling (Algorithm \ref{alg:block_coupling}), then the expected
  distance of $X'$ and $Y'$ satisfies
  $$\mathbb{E}[d(X', Y')] \leq\beta.$$
\end{lemma}
\begin{proof}
  Let $Y(v) = X(v) + 1$, hence, $Y(w) = X(w)$ for all $w\neq v$. Either the coupling
  is inactive due to $p > \frac{1}{2}$ or some block
  $B\in\BB$ is chosen at random.
  The probabilities in the following three cases are conditioned on $p \leq \frac{1}{2}$.

  \textit{Case I:} $v\in B$. Then, $X$ and $Y$ are equal on $V\setminus B$,
  in particular on $\partial B$. So the same admissible filling
  $\varphi\in\Omega_B$ is chosen for both $X$ and $Y$.
  Then~$X'= \REP X + \varphi!$ and $Y'= \REP Y + \varphi!$ are identical
  and~$d(X', Y') = 0$. This case happens with
  probability~$\#\{ B\in\BB \;\mid\; v\in B
  \}/\lvert\BB\rvert$.

  \textit{Case II:} $v\in\partial B$. For each block $B\in\BB$ with
  $v\in\partial B$, the probability for this to happen
  is~${1}/{\lvert \BB\rvert}$, and by definition of the
  block divergence $\mathbb{E}[d(X', Y')] \leq E_{B, v}$.

  \textit{Case III:} $v\notin (B\cup \partial B)$. As in case I, the same
  admissible filling $\varphi$ is sampled uniformly from~\mbox{$\Omega_{B\vert X}=\Omega_{B\vert Y}$}, but $Y'(v) = X'(v) + 1$ is being
  preserved, so~$d(X', Y') = 1$.
    
Putting all cases together gives
\begin{align*}
  \mathbb{E}[d(X', Y')]&\leq \frac{1}{2} + \frac{1}{2}\Bigg[
     \frac{\#\{ B\in\BB \mid v\in B \}}{\lvert\BB\rvert}\cdot 0 +
     \frac{1}{\lvert\BB\rvert}%
       \sum_{B\in\BB\hspace{0.1cm}\mid\hspace{0.1cm}v\in \partial B} E_{B, v}\\
  & + \left(1-\frac{\#\{ B\in\BB \mid v\in B \}}{\lvert\BB\rvert}
    -\frac{\#\{B\in\BB \mid v\in\partial B\}}{\lvert\BB\rvert}\right)
    \Bigg]\\
 & = 1-\frac{1}{2\lvert\BB\rvert}
   \left(\#\{ B\in\BB \hspace{0.1cm}\mid\hspace{0.1cm} v\in B\}
   \hspace{0.2cm} -
   \sum_{B\in\BB\hspace{0.1cm}\mid\hspace{0.1cm}v\in \partial B} (E_{B, v}-1)\right)\\
 &\leq \beta
\end{align*}
\vskip-5mm
\end{proof}

\begin{proof}[Proof of Theorem \ref{theorem:main_result}]
  We bound the mixing time of $\MM_{\BB}$ using block
  coupling. Lemma~\ref{lemma:one_step_block_coupling_distance} together with
  the assumptions in the theorem imply
  $$\mathbb{E}[d(X_{t+1}, Y_{t+1})] \leq \beta\cdot\mathbb{E}[d(X_t, Y_t)]$$
  when $(X_t, Y_t)$ is a cover relation. For a pair $(X,Y)\in\Omega\times\Omega$
  with $X\leq Y$ one can find a path $X = x_0, \ldots, x_r = Y$, where each
  $(x_i, x_{i+1})$ is a cover relation. Therefore, path coupling
  (Theorem~\ref{theorem:path_coupling}) yields a coupling of
  $\MM_{\BB}$ on all such pairs with the property
  $$\mathbb{E}[d(X_{t+1}, Y_{t+1})] \leq \beta\cdot\mathbb{E}[d(X_t,
  Y_t)]\hspace{0.1cm}.$$ 
  With Theorem~\ref{theorem:montone_coupling_classical} we then get
  $$\tau_{\MM_{\BB}}(\varepsilon)\leq
  \frac{\log\left(\frac{D}{\varepsilon}\right)}{1-\beta}\hspace{0.1cm}.
  $$
  To bound the mixing time of the up/down Markov chain
  $\MM$, we want to use the comparison technique from
  Theorem~\ref{theorem:comparison_technique}. Hence, we need a bound
  on the value of the $A$ in the theorem.

  Let $b := \max \{ \lvert B\rvert \mid B\in\BB\}$ and let
  $X \mapsto Y$ be a transition of the block Markov chain, i.e., 
  $(X, Y)\in E(\MM_{\BB})$. There is a block~$B$
  so that the~$k$-heights~$X$ and~$Y$ differ only on~$B$. By
  Lemma~\ref{lemma:path_lengths}, there is a path~$\gamma_{X, Y}$ from~$X$
  to~$Y$ of
  length~$\lvert \gamma_{X, Y}\rvert = d(X, Y) \leq k\cdot \lvert B\rvert \leq
  k\cdot b$ consisting of transitions
  in~$E(\MM)$. Choosing~$\gamma_{X, Y}$ as paths of length~$d(X, Y)$
  guarantees that along~$\gamma_{X, Y}$ only values of vertices in~$B$ are
  changed.

  For~$X, Y\in \Omega$ and $B\in\BB$,
  let~$\MM_{\BB}(X, Y\vert B)$ denote the transition
  probability of~$\MM_{\BB}$ for moving from state~$X$ to
  state~$Y$ given that in the first line of
  Algorithm~\ref{alg:block_chain_transition} block~$B$ is chosen and
  $p\leq\frac{1}{2}$. For $X\neq Y$ by the law of total probability we have
\begin{align}\label{equation_proof_main_result_I}
\MM_{\BB}(X, Y)=
  \frac{1}{2\lvert\BB\rvert}\sum_{B\in\BB}
  \MM_{\BB}(X, Y\vert B)
  \hspace{0.1cm}.
\end{align}

Now, fix some $(Q, R)\in E(\MM)$ which is not a loop and let~$v$ be the vertex on
which~$Q$ and~$R$ differ. The probability of this transition
is $$\MM(Q, R) = \frac{1}{4\lvert V\rvert}\hspace{0.1cm}.$$ As in the
statement of Theorem~\ref{theorem:comparison_technique}, consider the
set
$$\Gamma(Q, R) := \left\{ (X, Y)\in E(\MM_{\BB})
  \hspace{0.1cm}\vert\hspace{0.1cm} (Q,R) \in \gamma_{X, Y} \right\}$$ of
transitions in $\MM_{\BB}$ whose path in $\MM$ uses $(Q, R)$.
Let~$(X, Y)\in\Gamma(Q, R)$ and let~$B$ be a block for which
$\MM_{\BB}(X, Y\vert B) > 0$. Then the~$k$-heights~$X$ and~$Y$ can only differ
on vertices in~$B$. Moreover, since~$\gamma_{X, Y}$ uses~$(Q, R)$, we must
have~$X(v)\neq Y(v)$ and therefore $v\in B$. For an arbitrary block $B\in\BB$
observe
\begin{align}\label{equation_proof_main_result_II}
  \sum_{(X,Y)\in\Gamma(Q, R)}\MM_{\BB}(X, Y\vert B)\leq
  \sum_{X', Y'\in\Omega_B}\frac{1}{\lvert\Omega_B\rvert}=\lvert\Omega_B\rvert\leq (k+1)^b,
\end{align}
and the left hand side of (\ref{equation_proof_main_result_II}) equals zero if
$v\notin B$. By (\ref{equation_proof_main_result_I}) and
(\ref{equation_proof_main_result_II}) we get
\begin{align}
    \sum_{(X,Y)\in\Gamma(Q, R)}\MM_{\BB}(X, Y)&\leq\frac{1}{2\lvert\BB\rvert}\sum_{B\in\BB, v\in B}\sum_{(X, Y)\in\Gamma(Q,R)}\MM_{\BB}(X, Y\vert B)\\
    &\leq \frac{1}{2\lvert\BB\rvert}\#\{ B\in\BB\;\mid\;v\in B\}(k+1)^b \;.
\end{align}
Now consider
\begin{align*}
  A_{Q,R} &:=\frac{1}{\pi(Q)\MM(Q, R)}\sum_{(X, Y)\in\Gamma(Q, R)}\lvert\gamma_{X,Y}\rvert\cdot\pi(X)\cdot\MM_{\BB}(X, Y)\;.
\end{align*}
Since $\pi(Q)=\pi(X)$, and $\MM(Q, R) = 1 / (4 \lvert V\rvert )$ and
$\lvert\gamma_{X,Y}\rvert\leq kb$ we get
\begin{align*}
  A_{Q,R}  &\leq 4kb\lvert V\rvert \sum_{(X, Y)\in\Gamma(Q, R)}\MM_{\BB}(X, Y)\\
    &\leq \frac{2bk(k+1)^b\cdot\lvert V\rvert}{\lvert\BB\rvert}\cdot\#\{ B\in\BB\;\mid\;v\in B\}
\end{align*}
By using
$\pi_{\text{min}} = \frac{1}{\lvert\Omega\rvert} \geq \frac{1}{(k+1)^{\lvert
    V\rvert}}$ and $A := \max_{(Q, R)\in E(\MM)} A_{Q,R}$ in
Theorem~\ref{theorem:comparison_technique}, as well as 
\mbox{$m=\max \#\{ B\in\BB\;\mid\;v\in B\}$} and $D=\max d(X,Y) \leq k\lvert V\rvert$
we obtain the result:
\begin{align*}
  \tau_{\MM}(\varepsilon) &
  \leq \frac{4\log\left(\frac{1}{\varepsilon\cdot\pi_{\text{min}}}\right)}
  {\log\left(\frac{1}{2\varepsilon}\right)}\cdot A\cdot\tau_{\MM_{\BB}}(\varepsilon)\\
  &\leq \frac{4\log\left(\frac{(k+1)^{\lvert V\rvert}}{\varepsilon}\right)}
    {\log\left(\frac{1}{2\varepsilon}\right)}\cdot\frac{2bk(k+1)^b\cdot
    \lvert V\rvert}{\lvert\BB\rvert}\cdot
    \#\{ B\in\BB\;\mid\;v\in B\}\cdot
    \frac{\log\left(\frac{D}{\varepsilon}\right)}{1-\beta}\\
  &\leq \underbrace{\frac{8 bmk(k+1)^b}{(1-\beta)\lvert\BB\rvert}}_{c_{\BB, k}}\cdot 
    \frac{\left(\left(\log(\frac{1}{\varepsilon})\cdot
    \lvert V\rvert\right)+{\lvert V\rvert}^2 \cdot
    \log(k+1)\right)\cdot
    \log\left(\frac{k\lvert V\rvert}{\varepsilon}\right)}{\log(\frac{1}{2\varepsilon})}
\end{align*}
\end{proof}

\section{Applications}\label{section:applications}

In this section we will present how Theorem~\ref{theorem:main_result} and
Corollary~\ref{cor:main_result_corollary} can be applied on grid like graphs
and $3$-regular graphs. The results are based on computations of the block
divergence. The relevant data for $3$-regular graphs (Subsection~\ref{subsection:regular_graphs}) can be found in Table~\ref{tbl:block_divergence_type1} and Table~\ref{tbl:block_divergence_type2} in~Appendix~\ref{appendix:block_divergence_data}. The code used for all computations is available in~\cite{supplemental_data}.

\subsection{\texorpdfstring{$k$-heights}{k-heights} on toroidal rectangular grid graphs}\label{section:rectangular_grid_graphs}

For fixed~$g, h\in\mathbb{N}$ we consider the toroidal rectangular grid
graph~$G \sim (\mathbb{Z}/g\mathbb{Z})\times(\mathbb{Z}/h\mathbb{Z})$, i.e.,
the vertices are integer points~$(x, y)$ with horizontal
edges~$\left\{(x, y), (x+1, y)\right\}$ and vertical
edges~$\left\{(x, y), (x, y+1)\right\}$, where we take the~$x$- and
the~$y$-coordinate modulo~$g$ and modulo~$h$, respectively; see
Figure~\ref{fig:rectangular_grid}.

\begin{figure}[htb]
\centering
\includegraphics{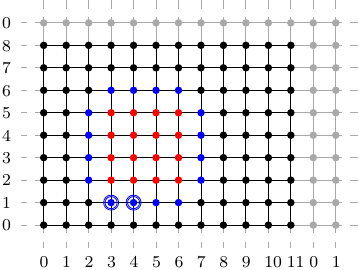}
\caption{Toroidal rectangular grid graph of size $12\times 9$.}
\label{fig:rectangular_grid}
\end{figure}

For showing that the up/down Markov chain on~$2$-heights or~$3$-heights of~$G$
is rapidly mixing, we use the family~$\BB$ consisting of all contiguous~$4\times 4$ blocks. In Figure~\ref{fig:rectangular_grid},
such a block is highlighted in red. Assuming that~$g, h$ are sufficiently
large, there are exactly~$\lvert\BB\rvert = g\cdot h$ such blocks, each vertex
is contained in exactly~$16$ blocks and forms part of~$16$ block
boundaries. Aiming for an application of
Corollary~\ref{cor:main_result_corollary}, we have to bound the block
divergence. We do this on the basis of of massive computations.

The boundary $\partial B$ consists of four paths of length~$4$ (blue vertices
in Figure~\ref{fig:rectangular_grid}). In every boundary constraint
$X\in\Omega_{\partial B}$, the values of two successive vertices on one of
these paths can differ by at most~$1$. Further, in an extensible boundary
constraint, the values of the last vertex of one path and the first vertex of
the next path differ by at most~$2$, because otherwise there is no possible
value for the vertex in the corner of $B$.

If we consider the boundary as a chain of~$16$ transitions, we can compute the
number of the extensible boundary constraints as
\[
  \operatorname{tr}(Q_kP_k^3Q_kP_k^3Q_kP_k^3Q_kP_k^3) =
  \operatorname{tr}\left((Q_kP_k^3)^4\right),\] where $P_k$ and $Q_k$ are both
matrices of size $k\times k$ with \[ (P_k)_{i,j} :=
\begin{cases}
   1 & \text{if }\lvert i - j \rvert \leq 1 \\
   0 & \text{otherwise}
\end{cases}\text{ \hspace{1.1cm} }
 (Q_k)_{i,j} := \begin{cases}
        1 & \text{if }\lvert i - j \rvert \leq 2 \\
        0 & \text{otherwise}
\end{cases}.
\]

If $k=2$ this gives~$2.825.761$ extensible boundary constraints;
if~\mbox{$k=3$} we get~$15.784.802$ extensible boundary constraints. For any
block~\mbox{$B\in\BB$} and~\mbox{$v\in\partial B$}, when we want to
compute~$E_{B, v}$, Lemma~\ref{lemma:linearity_expectation} tells us that this
is the ma\-xi\-mum of~$\mathbb{E}[w(u)] - \mathbb{E}[w(\ell)]$
with random admissible fillings~\mbox{$\ell \sim \mathcal{U}(\Omega_{B\vert X})$}
and~\mbox{$u \sim \mathcal{U}(\Omega_{B\vert Y})$} for two $k$-heights $X, Y$ that
differ in a single vertex~$v\in\partial B$,~\mbox{$X(v) = Y(v)-1$}.  For
symmetry reasons, it is sufficient for the maximization to compute~$E_{B,v}$
for the two vertices~$v\in\partial B$ which are encircled in
Figure~\ref{fig:rectangular_grid}.

Given a boundary constraint $X\in\Omega_{\partial B}$ we compute
$\mathbb{E}[w(\ell)]$ for $\ell \sim \mathcal{U}(\Omega_{B\vert X})$ with a
dynamic programming approach. For each row of the block we consider all the
fillings consistent with the boundary conditions as vectors.  By going from
row to row we compute for each vector the number and total weight of all consistent assignments of vectors to previous rows. This then allows to compute the total
weight~$\sum_{\ell\in\Omega_{B\vert X}} w(\ell)$ of all admissible fillings
and their number~$\lvert\Omega_{B\vert X}\rvert$. The value of
$\mathbb{E}[w(\ell)]$ is the quotient of the two numbers.

We do this computation for each boundary constraint~$X\in\Omega_{\partial B}$ and store the result. Next, we iterate over all cover relations $X, Y\in\Omega_{\partial B}$ (up to symmetry) in order to compute~\mbox{$\max \{ E_{B, v} : v\in\partial B \}$}. The results are shown in Table~\ref{tbl:block_divergence_rectangular_grids}. In the case $k=4$, we interrupted the execution but had already found a cover relation that gives the lower bound in the table.

\begin{table}[ht]
    \centering
    \begin{tabular}{ c|c } 
        $k$ & $\max \{ E_{B, v} : v\in\partial B \}$ \\
        \hline
        $2$ & $\approx 1.225092$ \\ 
        $3$ & $\approx 1.752678$ \\
        $4$ & $> 2.27$
    \end{tabular}
    \caption{Block divergence for blocks of size $4\times 4$ in rectangular grid.}
    \label{tbl:block_divergence_rectangular_grids}
\end{table}

\begin{theorem}
  \label{theorem:rapidly_mixing_toroidal_rectangle_graph}
  Let $G=(V, E)$ be a toroidal rectangular grid
  graph,~\mbox{$n = \lvert V\rvert$}. For~$k\in\{2, 3\}$ the up/down
  Markov-chain $\MM$ operating on $k$-heights of $G$ is rapidly
  mixing. More precisely, the mixing time is upper bounded by
  $$\tau(\varepsilon)<
  c_k\cdot\frac{\left(\left(\log(\frac{1}{\varepsilon})\cdot n\right)+n^2
      \cdot\log(k+1)\right)\cdot\log\left(\frac{kn}{\varepsilon}\right)}
  {\log(\frac{1}{2\varepsilon})}\in\mathcal{O}\left(n^2\log n\right),
  $$
  where $c_2 = 2.844202\cdot 10^{10}$ and $c_3 = 1.333706\cdot 10^{13}$.
\end{theorem}

\begin{proof}
In the notation of Corollary~\ref{cor:main_result_corollary}, for the family
$\BB$ of all contiguous blocks of size $4\times 4$ we have $\check{m}=s=16$. We
obtain
\[
  1-\frac{1}{2\lvert\BB\rvert}\left(\check{m}-s\cdot (E_\text{max} - 1)\right) =:
  \beta_k < 1
\]
if and only if $E_\text{max} < 2$. This holds for $k\in\{2, 3\}$ as seen by
the computational results in
Table~\ref{tbl:block_divergence_rectangular_grids}. Hence,
Corollary~\ref{cor:main_result_corollary} implies that the up/down Markov
chain is rapidly mixing in these two cases and that the mixing times are upper
bounded as in Theorem~\ref{theorem:main_result}.

In the case $k=3$, we know $E_\text{max} < 1.752678$, from which we get
\[
  \beta_3 = 1-\frac{1}{2\lvert\BB\rvert}\left(\check{m}-s\cdot (E_\text{max} - 1)\right)
      < 1 - \frac{1.978576}{\lvert\BB\rvert},
\]
and in the notation as in Theorem~\ref{theorem:main_result}, using $b=m=16$,
we obtain
\[ c_{\BB, 3} = \frac{8\cdot bmk(k+1)^b}{(1-\beta_3)\lvert\BB\rvert}
     < \frac{8\cdot bmk(k+1)^b}{1.978568} < 1.333706\cdot 10^{13}.
\]

In the case $k=2$, we do the same calculation based on
$E_\text{max} < 1.225093$ and obtain $c_{\BB, 2} < 2.844202\cdot 10^{10}$.
\end{proof}

For $k=4$ our computation shows $E_\text{max} > 2$
(Table~\ref{tbl:block_divergence_rectangular_grids}), which means that in
Corollary~\ref{cor:main_result_corollary} we cannot find such a $\beta_4 <
1$. We believe that the up/down Markov chain is rapidly mixing independently
of $k$. We tried to use larger blocks for~$\BB$ but ran out of computational
power.

\subsection{\texorpdfstring{$k$-heights}{k-heights}
  on toroidal hexagonal grid graphs}
\label{ssec:mixingtoroidalhexagonalgraph}

For fixed~$g,h\in\mathbb{N}$ we consider the triangular grid graph
on~$(\mathbb{Z}/g\mathbb{Z})\times (\mathbb{Z}/h\mathbb{Z})$, i.e., the
vertices are points~$(x, y)\in\mathbb{Z}\times\mathbb{Z}$ with horizontal
edges~$\left\{(x, y), (x+1, y)\right\}$, vertical
edges~$\left\{(x, y), (x, y+1)\right\}$ and diagonal
edges~$\left\{(x, y), (x+1, y+1)\right\}$, and we consider the~$x$-
and~$y$-coordinates modulo~$g$ and modulo~$h$; respectively. Of the so
obtained plane graph we take the dual graph~$G$, whose vertices correspond to
the~$2\cdot g\cdot h$ triangular faces and edges in~$G$ correspond to pairs of
triangles that share a side; see Figure~\ref{fig:hexagonal_grid}. The
graph~$G$ is a hexagonal grid.

\begin{figure}[htb]
  \centering
  \begin{subfigure}[c]{.3\textwidth}
      \centering
      \includegraphics[page=2]{figures/hexagonal_grid.pdf}
      \caption{}
      \label{fig:hexagonal_grid:primal_graph}
  \end{subfigure}
  \begin{subfigure}[c]{.3\textwidth}
      \centering
      \includegraphics[page=1]{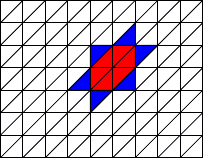}
      \caption{}
      \label{fig:hexagonal_grid:dual_graph}
  \end{subfigure}
  \caption{Hexagonal grid (a) or dual (b) of toroidal triangle grid of size $10\times 8$. Red vertices or faces form a block; blue vertices or faces are part of the boundary~$\partial B$.}
  \label{fig:hexagonal_grid}
\end{figure}

Every point~$(x, y)\in\mathbb{Z}\times\mathbb{Z}$ defines a block $B_{x,y}$
consisting of~$6$ vertices that correspond to the triangular faces 
incident to the vertex $(x, y)$; one such block is indicated in red in
Figure~\ref{fig:hexagonal_grid}. We use the family
$\BB := \{ B_{x, y} : x, y\in\mathbb{Z}\}$, which for sufficiently large $g,h$
has $g\cdot h$ blocks, each vertex is contained in~$3$ blocks and~$3$ block
boundaries. Using the matrix~$P_k$ defined
in~Section~\ref{section:rectangular_grid_graphs}, we can compute the number of
$k$-heights on a block as~$\lvert\Omega_B\rvert = \operatorname{tr}(P_k^6)$.

For every block~$B$, the boundary~$\partial B$ consists of~$6$ isolated
vertices. Therefore, there are exactly~\mbox{$\lvert\Omega_{\partial B}\rvert = (k+1)^6$}
legal boundary constraints. If~$k\geq 4$, some of them are not extensible, hence we
could skip them in our computation. We iterated over all
boundary constraints~$\Omega_{\partial B}$ and fillings $\Omega_B$, as their
cardinalities are small; see
Table~\ref{tbl:block_divergence_hexagonal_grids}. The computation detects non
extensible boundary constraints and simply ignores them. For symmetry reasons,
when computing~$\mathbb{E}[w(u)] - \mathbb{E}[w(\ell)]$
where~$\ell \sim \mathcal{U}(\Omega_{B\vert X})$
and~$u \sim \mathcal{U}(\Omega_{B\vert Y})$, we only need to do this for cover
relations $(X, Y)$ on~$\partial B$ that differ in a fixed
vertex~$v\in\partial B$.

\begin{table}[ht]
\centering
 \begin{tabular}{ c|c|c|c } 
    $k$ & $\lvert\Omega_B\rvert$ & $\lvert\Omega_{\partial B}\rvert$ &
         $\max \{ E_{B, v} : v\in\partial B \}$ \\
    \hline
    $2$ & 199 & 729 & $\approx$ 0.798658\\ 
    $3$ & 340 & 4096 & $\approx$ 1.831905 \\
    $4$ & 481 & 15625 & $\approx$ 2.892857 \\
    $5$ & 622 & 46656 & 3.0 \\
    $6$ & 763 & 117649 & 3.0
 \end{tabular}
\caption{Block divergence for blocks of~$6$ hexagonal shaped vertices.}
\label{tbl:block_divergence_hexagonal_grids}
\end{table}

\begin{theorem}
    \label{theorem:rapidly_mixing_toroidal_hexagonal_graph}
    Let $G=(V, E)$ be a toroidal hexagonal grid graph, $n = \lvert
    V\rvert$. For~$k\in\{2, 3\}$ the up/down Markov-chain $\MM$
    operating on $k$-heights of $G$ is rapidly mixing. More precisely, the
    mixing time is upper bounded by
    $$\tau(\varepsilon)<
    c_k\cdot\frac{\left(\left(\log(\frac{1}{\varepsilon})\cdot n\right)+n^2
      \cdot\log(k+1)\right)\cdot\log\left(\frac{kn}{\varepsilon}\right)}
    {\log(\frac{1}{2\varepsilon})}\in\mathcal{O}\left(n^2\log n\right),
    $$
    where $c_2 = 1.165099\cdot 10^{5}$ and $c_3 = 7.017788\cdot 10^{6}$.
\end{theorem}

\begin{proof}
  In the notation of Corollary~\ref{cor:main_result_corollary}, for the family
  of blocks $\BB$ that we have described above we have $\check{m}=s=3$. We obtain
  \[ 1-\frac{1}{2\lvert\BB\rvert}\left(\check{m}-s\cdot (E_\text{max} - 1)\right)
    =: \beta_k < 1
  \]
  if and only if $E_\text{max} < 2$. This holds for $k\in\{2, 3\}$ as seen by
  the computational results in~Table~\ref{tbl:block_divergence_rectangular_grids}. Hence,
  Corollary~\ref{cor:main_result_corollary} implies that the up/down Markov
  chain is rapidly mixing in these two cases and that the mixing times are
  upper bounded as in Theorem~\ref{theorem:main_result}.

  In the case $k=3$, we know $E_\text{max} < 1.831906$, from which we get
  \[ \beta_3 = 1-\frac{1}{2\lvert\BB\rvert}
                                \left(\check{m}-s\cdot (E_\text{max} - 1)\right)
    < 1 - \frac{0.252141}{\lvert\BB\rvert},
  \]
  and in the notation as in Theorem~\ref{theorem:main_result}, using $b=6$ and $m=3$ we obtain
  \[
    c_{\BB, 3} = \frac{8\cdot bmk(k+1)^b}{(1-\beta_3)\lvert\BB\rvert}
     < \frac{8\cdot bmk(k+1)^b}{0.252141} < 7.017788\cdot 10^{6}.
  \]
  In the case $k=2$, we do the same calculation based on
  $E_\text{max} < 0.798659$ and obtain the bound \mbox{$c_{\BB, 2} < 1.165099\cdot 10^{5}$}.
\end{proof}

Again, for $k=4$ our computation shows $E_\text{max} > 2$
(Table~\ref{tbl:block_divergence_hexagonal_grids}), which means that in
Corollary~\ref{cor:main_result_corollary} we cannot find such a $\beta_4 < 1$.

\subsection{\texorpdfstring{$k$-heights}{k-heights}
  on planar \texorpdfstring{$3$-regular}{3-regular} graphs}
\label{subsection:regular_graphs}

\begin{figure}[htb]
\centering
\includegraphics{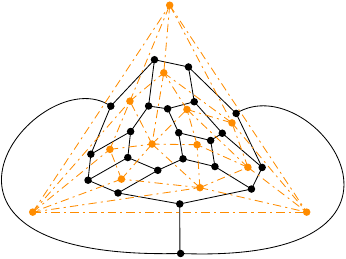}
\caption{Example of a dual graph (black, solid) of a $4$-connected triangulation (orange, dashed)}
\label{fig:dual_4_con_triang}
\end{figure}

\begin{theorem}
\label{theorem:rapidly_mixing_regular_graph}
Let $G=(V, E)$ be a simple $2$-connected $3$-regular planar graph, and
let~\mbox{$n = \lvert V\rvert$}.

\begin{itemize}
    \item[(1)] The mixing time~$\tau(\varepsilon)$ of the up/down Markov chain $\MM$ operating on $2$-heights is upper bounded by \[
        \tau(\varepsilon)< c_k\cdot\frac{\left(\left(\log(\frac{1}{\varepsilon})\cdot n\right)+n^2 \cdot\log(k+1)\right)\cdot\log\left(\frac{kn}{\varepsilon}\right)}{\log(\frac{1}{2\varepsilon})}\in\mathcal{O}\left(n^2\log n\right)\;,
    \] where~$c_2 = 4.391132\cdot 10^7$.
    \item[(2)] If~$G$ is even $3$-connected, then, for $k\in\{2, 3\}$, the mixing time~$\tau(\varepsilon)$ of the up/down Markov chain $\MM$ operating on $k$-heights is upper bounded by the same expression as in (1) with constants $c_2 = 2.195097\cdot 10^7$ and $c_3 = 4.852027 \cdot 10^9$.
    \item[(3)] If~$G$ is the dual graph of a $4$-connected triangulation, then, for $k\in\{2, 3\}$, the mixing time~$\tau(\varepsilon)$ of the up/down Markov chain $\MM$ operating on $k$-heights is upper bounded by the same expression as in (1) with constants $c_2 = 1.489256\cdot 10^7$ and $c_3 = 4.852027\cdot 10^9$.
\end{itemize}
\end{theorem}

For a fixed plane embedding of~$G$ we construct a family~$\BB$ of blocks of
the following two types. For every face~$F$ of degree~$d \leq 10$ we consider
the block consisting of all~$d$ boundary vertices and call it a block of
type~$1$, or, more specifically, a block of type~$1_d$; see
Figure~\ref{fig:block_two_types:type_1}. In~$\BB$ we include $8$ identical
copies of each such block. For every face~$F$ of degree~$d > 10$, we consider
all sets of~$8$ successive vertices on the boundary of~$F$ and call them
blocks of type~$2$. We include each of these~$d$ blocks (a single time)
in~$\BB$; see Figure~\ref{fig:block_two_types:type_2} for an example.

\begin{figure}[htb]
  \centering
  \begin{subfigure}[b]{.49\textwidth}
      \centering
      \includegraphics[page=1]{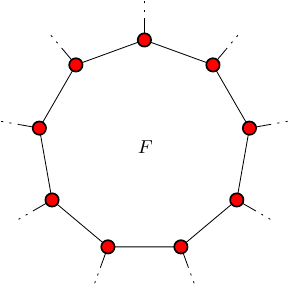}
      \caption{}
      \label{fig:block_two_types:type_1}
  \end{subfigure}
  \hfill
  \begin{subfigure}[b]{.49\textwidth}
      \centering
      \includegraphics[page=2]{figures/block_two_types.pdf}
      \caption{}
      \label{fig:block_two_types:type_2}
  \end{subfigure}
  \caption{Red vertices form a block of type $1_9$ in (a)
    and a block of type $2$ in~(b). Blue vertices are part
    of the boundary~$\partial B$.}
  \label{fig:block_two_types}
\end{figure}

As a first consequence of this construction, each vertex~$v$ is contained in
exactly~$24$ blocks in~$\BB$: As $\operatorname{deg}(v)=3$ and~$G$ is
$2$-connected, vertex $v$ belongs to~$3$ distinct faces, each of which
contributes~$8$ blocks containing $v$.

When computing the block divergence $E_{B, v}$ for some $v\in\partial B$, we
distinguish different cases depending on the type of~$B$ and the adjacency
relations between~$B$ and~$v$. If~$B$ is of type~$1_d$, we label the vertices in
clockwise order around~$F$ as~$1, \ldots, d$. In fact, due to symmetry, it
will not be relevant at which vertex the numbering starts. If~$B$ is of
type~$2$, we label its vertices in clockwise order as~$1, \ldots, 8$.

Note that~$v\in\partial B$ can be adjacent to one, two or three vertices
of~$B$. The cases in the following case distinction will be tagged by the type
of~$B$ followed by the labels of the neighbors of~$v$ in square brackets. For
instance, Case~$1_5[1]$ describes the scenario in which~$v$ is adjacent to a
single vertex of a block of type~$1_5$. Due to symmetry, the cases 
$1_5[x]$ with $x\in \{1,\ldots,5\}$ are all equivalent
to Case~$1_5[1]$; they result in the same block divergence~$E_{B, v}$. As
another example, Figure~\ref{fig:block_divergence_cases} shows Case~$2[2,5]$,
i.e.,~$v$ is adjacent to vertices~$2$ and~$5$ of a block of type~$2$. This
case is symmetrical to Case~$2[4,7]$, both result in the same block
divergence~$E_{B, v}$.

\begin{figure}[htb]
\centering
\includegraphics{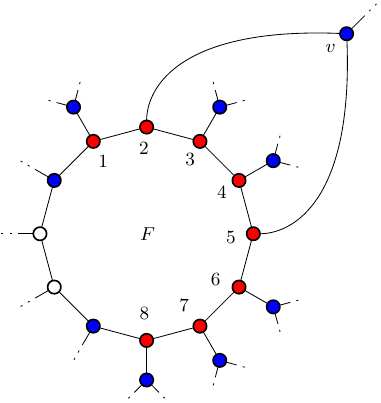}
\caption{Case~$2[2,5]$: Vertex~$v$ is adjacent to vertices~$2$ and~$5$ from a
  block~$B$ which is of type~$2$. The block divergence $E_{B, v}$ is the same
  as in Case~$2[4,7]$.}
\label{fig:block_divergence_cases}
\end{figure}

In a block~$B$ of type~$2$, vertices~$1$ and~$8$ have two neighbors
in~$\partial B$. Note that when~$v$ is one of these two neighbors, our
tag will not tell whether $v$ belongs to the boundary of~$F$. This
information would not affect~$E_{B,v}$.

Further, note that for a block~$B$ of any type, so far we did not consider the
cases in which some boundary vertex other than~$v$ is adjacent to multiple
vertices in~$B$. For instance, in the example in
Figure~\ref{fig:block_divergence_cases}, vertices~$3$ and~$4$ or some of the
vertices~$1$,~$6$,~$7$,~$8$ could be adjacent to a common boundary
vertex. However, these cases are already covered in our computation
of~$E_{B, v}$, where we consider all boundary constraints
$\Omega_{\partial B}$, including those where some of the boundary vertices are
assigned the same value. Hence, these further cases can only have a smaller
block divergence. 

For $k=2$ and $k=3$, we used a computer program to compute $E_{B, v}$ for all
cases described above; see Table~\ref{tbl:block_divergence_type1} for cases
where~$B$ is of type~$1$ and Table~\ref{tbl:block_divergence_type2} for cases
where~$B$ is of type~$2$. Whenever multiple cases result in the same block
divergence $E_{B, v}$ for symmetry reasons, the tables contain only one of them.

It turns out that in this setting Corollary~\ref{cor:main_result_corollary} is not strong enough for proving that~$\MM$ is rapidly mixing: Each vertex is contained in exactly~$\check{m}=24$ blocks and in at most~$s = 30$ boundaries of blocks; and even in the case~$k=2$ we have~$E_{\textup{max}} \approx 2.367241$ (with Case~$1_{10}[1,3,7]$ being the extremal case; see Table~\ref{tbl:block_divergence_type1} and Table~\ref{tbl:block_divergence_type2}), which gives $m-s(E_{\text{max}} - 1) < -17.0172 < 0$. Therefore, instead of applying Corollary~\ref{cor:main_result_corollary}, we will apply Theorem~\ref{theorem:main_result} directly.

\begin{lemma}\label{lemma:block_divergence_2connected}  
  If~$G=(V, E)$ is a simple $2$-connected $3$-regular planar graph, then, with
  the family~$\BB$ of blocks described above and~$k=2$, we
  have
  \[ \#\{ B\in\BB \;\mid\; v\in B\}\hspace{0.2cm} - \sum_{B\in\BB\;\mid\; v\in
      \partial B} (E_{B, v}-1)\hspace{0.1cm}>\hspace{0.1cm} 10.32755
    \] for all $v\in V$.
\end{lemma}

\begin{proof}
  We have already observed $\#\left\{ B\in\BB\;\mid\; v\in B\right\}=24$. For
  analysing the sum on the right hand side, we have to consider all blocks
  $B\in\BB$ with~\mbox{$v\in\partial B$} and look up the corresponding cases
  in Table~\ref{tbl:block_divergence_type1} and
  Table~\ref{tbl:block_divergence_type2}. As~$v$ has degree~$3$ and~$G$ is
  $2$-connected,~$v$ is incident to exactly three pairwise distinct
  faces~$H_1$,~$H_2$ and~$H_3$. Its three neighbors~$w_1$,~$w_2$ and~$w_3$ are
  incident to three further faces~$F_1$,~$F_2$ and~$F3$, as shown in
  Figure~\ref{fig:vertex_incident_faces}.

\begin{figure}[htb]
\centering
\includegraphics{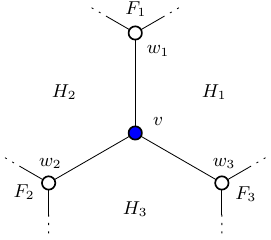}
\caption{A vertex~$v$ can only be in the boundary of a block
  induced from any of the faces $H_1$, $H_2$, $H_3$, $F_1$, $F_2$ and $F_3$.}
\label{fig:vertex_incident_faces}
\end{figure}

There are at most~$30$ blocks~$B\in\BB$ with $v\in\partial B$: At most~$8$
induced by each of the faces~$F_1$,~$F_2$ and~$F_3$, and at most~$2$ induced
by each of the faces~$H_1$,~$H_2$ and $H_3$.

In the case in which the faces $H_1$,~$H_2$,~$H_3$,~$F_1$,~$F_2$,~$F_3$ are
pairwise distinct, every block~$B\in\BB$ contains at most one vertex adjacent
to~$v$. In the corresponding cases in Table~\ref{tbl:block_divergence_type1}
and Table~\ref{tbl:block_divergence_type2}, we have $E_{B, v} < 0.80456$,
with~$1_{10}[1]$ being the extremal case. Moreover, there are at least~$24$
blocks~\mbox{$B\in\BB$} with~$v\in\partial B$, namely the~$8$ blocks induced by each of
the faces~$F_1$,~$F_2$ and~$F_3$. From this we directly get
\[
  \#\{ B\in\BB \;\mid\; v\in B\}\hspace{0.2cm} -
  \sum_{B\in\BB\;\mid\; v\in \partial B} (E_{B, v}-1)\quad
  > \quad 24 + 24\cdot 0.19544\quad > \quad 28 .
\]    
However, the faces $F_1$,~$F_2$ and~$F_3$ do not need to be pairwise distinct,
and they can be identical to a face $H_1$,~$H_2$ or~$H_3$. In such cases,
blocks may occur that contain two or three vertices adjacent to~$v$, which
requires a refined analysis.

From the data in Table~\ref{tbl:block_divergence_type1} and
Table~\ref{tbl:block_divergence_type2} we get
\[ E^* := \max_{B\in\BB\;\mid\;v\in\partial B} \hspace{0.2cm}
  \frac{E_{B, v} - 1}{\#\left\{w\in B\;:\; \{v, w\}\in E\right\}} < 0.455748,
\]
with Case $1_{10}[1,3,7]$ being the extremal case. Now we have:
\begin{align*}
    \sum_{B\in\BB\;\mid\; v\in\partial B} E_{B, v} - 1 &= \sum_{B\in\BB\;\mid\; v\in \partial B} \#\left\{w\in B : \{v, w\} \in E\right\}\frac{E_{B, v} - 1}{\#\left\{w\in B : \{v, w\} \in E\right\}}\\
        & \leq \hspace{0.1cm} E^* \cdot \sum_{i=1}^3 \#\left\{B\in\BB : w_i\in B, v\notin B \right\}\\
        & < \hspace{0.1cm} 0.455748\cdot 30 = 13.67245
\end{align*}

    In the last inequality we used that there are at most~$10$ blocks~$B\in \BB$ with~$w_i\in B$ and~$v\notin B$: At most~$8$ induced by~$F_i$ and at most one induced by each of the two faces~$H_j$ that are incident to~$w_i$. Finally, we obtain
    \begin{align*}
        &\#\{ B\in\BB \;\mid\; v\in B\}\hspace{0.2cm} - \sum_{B\in\BB\;\mid\; v\in \partial B} (E_{B, v}-1)\\
        &>\hspace{0.2cm}24 - 13.67245 \;=\; 10.32755\hspace{0.1cm}.
    \end{align*}
\end{proof}

\begin{lemma}\label{lemma:block_divergence_3connected}
    If~$G=(V, E)$ is a $3$-connected $3$-regular planar graph, then, with the family~$\BB$ of blocks constructed above and in the case~$k\in\{2,3\}$, we have \[
        \#\{ B\in\BB \;\mid\; v\in B\}\hspace{0.2cm} - \sum_{B\in\BB\;\mid\; v\in \partial B} (E_{B, v}-1)\hspace{0.1cm}>\hspace{0.1cm}\begin{cases}
            20.659512 & k = 2 \\
            2.489598 & k = 3
        \end{cases}
    \] for all $v\in V$. 
\end{lemma}
\begin{proof}
    The proof is similar to the proof of Lemma~\ref{lemma:block_divergence_2connected}, but the $3$-connectivity allows us to exclude cases that would imply a $2$-separator. For instance, in Case~$2[2,5]$, depicted in Figure~\ref{fig:block_divergence_cases}, the vertices~$2$ and~$5$ form a $2$-separator, hence, it cannot occur in~$G$. Indeed, we can exclude all cases in which~$v$ is adjacent to two non-consecutive vertices on~$B$ (For example, Case~$2[2,5]$) 
    or adjacent to three vertices on~$B$ (For example, Case~$2[2,4,5]$). The remaining cases are those in which~$v$ has only one neighbor in~$B$ (For example, Case~$2[3]$) or exactly two neighbors that lie consecutively on~$B$ (For example, Case~$2[3,4]$).

    First, we treat the case~$k=3$. For upper bounding~$E^*$ as in the proof of Lemma~\ref{lemma:block_divergence_2connected}, we maximize over the cases that we did not exclude. Using the values from Table~\ref{tbl:block_divergence_type1} and Table~\ref{tbl:block_divergence_type2}, we obtain~$E^* < 0.848707$ with Case~$1_{10}[1]$ being the extremal case.
    
    Recall the notation from Figure~\ref{fig:vertex_incident_faces}. From the $3$-connectivity of~$G$ it follows that each of the faces~$F_1$,~$F_2$ and~$F_3$ is distinct to the faces $\{H_1, H_2, H_3\}$: Clearly,~$F_1$ is distinct to $H_1$ and $H_2$, since otherwise, the vertex~$w_1$ would be a separator. Further,~$F_1$ is distinct to~$H_3$, as otherwise, the vertices~$\{v, w_1\}$ would separate~$w_3$ from~$w_2$. The statement for~$F_2$ and~$F_3$ follows by symmetry. 
    
    Let~$\BB_F \subset \BB$ be the set of blocks that are induced by the faces~$\{F_1, F_2, F_3\}$, and let~$\BB_H \subset \BB$ be the set of blocks that are induced by the faces~$\{H_1, H_2, H_3\}$. Using the fact that~$\{H_1, H_2, H_3\}$ are distinct from~$\{F_1, F_2, F_3\}$, we can get the following: The only way in which~$v$ can be in the boundary of a block~$B$ induced by~$H_i$ is when~$H_i$ has degree larger than~$10$ and~$B$ is of type~$2$ consisting of the~$8$ consecutive vertices either before or behind~$v$ on the boundary of~$H_i$. These are Case~$2[8]$ and Case~$2[1]$, which lead to the same block divergence~\mbox{$E_{B, v} \approx 1.190238$}; see Table~\ref{tbl:block_divergence_type2}. 
    
    This allows us to refine the analysis that we already did in the proof of Lemma~\ref{lemma:block_divergence_2connected}:   
    \begin{align*}
        & \sum_{B\in\BB\;\mid\; v\in\partial B} E_{B, v} - 1\\
        & = \left(\sum_{B\in\BB_H\;\mid\; v\in\partial B} E_{B, v} - 1\right) + \left(\sum_{B\in\BB_F\;\mid\; v\in\partial B} E_{B, v} - 1\right) \\
        & < \hspace{0.1cm} 6\cdot (1.190239 - 1) + E^* \cdot \sum_{i=1}^3 \#\left\{B\in\BB_F : w_i\in B, v\notin B \right\}\\
        & < \hspace{0.1cm}6\cdot 0.190239 +  0.848707\cdot 24 = 21.510402
    \end{align*}
    
    From this we obtain our result for the case $k=3$:
    \begin{align*}
        &\#\{ B\in\BB \;\mid\; v\in B\}\hspace{0.2cm} - \sum_{B\in\BB\;\mid\; v\in \partial B} (E_{B, v}-1)\\
        &>\hspace{0.2cm}24 - 21.510402 \;=\; 2.489598
    \end{align*}

    The result for $k=2$ is obtained using the very same calculations. From Table~\ref{tbl:block_divergence_type1} and Table~\ref{tbl:block_divergence_type2} we get $E^* < 0.212547$ (with Case~$1_{10}[1,2]$ being the extremal case), for Case~2[8] and Case~2[1] we have $E_{B, v} \approx 0.706599$ and obtain
    \begin{align*}
        &\#\{ B\in\BB \;\mid\; v\in B\}\hspace{0.2cm} - \sum_{B\in\BB\;\mid\; v\in \partial B} (E_{B, v}-1)\\
        &>\hspace{0.2cm}24 - (6\cdot (0.706560 - 1) + 0.212547\cdot 24) \;=\; 20.659512
    \end{align*}
\end{proof}

\begin{lemma}\label{lemma:block_divergence_dual_of_4con_triang}
    If~$G=(V, E)$ is the dual graph of a $4$-connected triangulation, then, with the family~$\BB$ of blocks constructed above and in the case~$k\in\{2,3\}$, we have \[
        \#\{ B\in\BB \;\mid\; v\in B\}\hspace{0.2cm} - \sum_{B\in\BB\;\mid\; v\in \partial B} (E_{B, v}-1)\hspace{0.1cm}>\hspace{0.1cm}\begin{cases}
            30.4512 & k = 2 \\
            2.489598 & k = 3
        \end{cases}
    \] for all $v\in V$. 
\end{lemma}
\begin{proof}
    We know that $G$ is $3$-connected as the dual of a $4$-connected triangulation, so we could directly refer to Lemma~\ref{lemma:block_divergence_3connected}. However, the fact that~$G$ is the dual of a $4$-connected triangulation allows to exclude further cases, namely all cases in which~$v$ has more than one neighbor in~$B$. Apart from this further restriction, the proof is exactly the same. We just mention the relevant data: From Table~\ref{tbl:block_divergence_type1} and Table~\ref{tbl:block_divergence_type2} we get for~$k=2$ the bound~\mbox{$E^* < -0.195440$} and for~$k=3$ the bound~$E^* < 0.848707$ (with Case~$1_{10}[1]$ being the extremal case for both~$k\in\{2, 3\}$). 
\end{proof}

\begin{proof}[Proof of Theorem~\ref{theorem:rapidly_mixing_regular_graph}]
    We apply Theorem~\ref{theorem:main_result}. In the case of planar $2$-heights on $2$-connected $3$-regular graphs, by Lemma~\ref{lemma:block_divergence_2connected}, we have
    \begin{align*}
        \beta_2 &:= 1-\frac{1}{2\lvert\BB\rvert}\left(\#\{ B\in\BB \;\mid\; v\in B\}\hspace{0.2cm} - \sum_{B\in\BB\;\mid\;v\in \partial B} (E_{B, v}-1)\right) \\
        & <\hspace{0.1cm} 1-\frac{5.163775}{\lvert\BB\rvert},
    \end{align*}
    and in the notation as in Theorem~\ref{theorem:main_result}, using $b=10$ and $m=24$, we obtain \[
        c_{\BB, 2} = \frac{8\cdot bmk(k+1)^b}{(1-\beta_2)\lvert\BB\rvert}< \frac{8\cdot bmk(k+1)^b}{5.163775} < 4.391132\cdot 10^{7}.
    \]

    Part (2) and (3) follow using the same calculation and Lemma~\ref{lemma:block_divergence_3connected} and Lemma~\ref{lemma:block_divergence_dual_of_4con_triang}.
\end{proof}

\section{Glauber dynamics}\label{sec:glauber_dynamics}

A \textit{spin system} consists of a graph~$G=(V, E)$, where the vertices~$V$ are
also called \textit{sites}, a finite set~$Q = \{1, \ldots, q\}$ of possible
spins and a set of \textit{feasible configurations}~$\Omega\subset Q^V$,
where~$Q^V := \{\varphi: V\to Q\}$. Examples include
proper~\mbox{$q$-colorings}, independent sets (hardcore model) and magnetic
states (Ising model);
see~\cite{blancaCaputoZongchenParisiStefankovicVigoda2022}, from which we take
the following notation.

Spin systems are inspired by physics; in a configuration~$\varphi\in Q^V$
there are interacting forces between neighbor sites depending on their spins,
making some of the configurations less likely or infeasible. This is expressed
by a \textit{weight} or "inverse energy"
function~$w: Q^V\to\mathbb{R}_{\geq 0}$,
$$
w(\varphi) = \prod_{\{v, w\}\in E} A_{vw}(\varphi(v), \varphi(w))
\hspace{0.2cm}\prod_{v\in V} B_v(\varphi(v))\;,
$$
where~$A_{vw}: Q\times Q \to\mathbb{R}_{\geq 0}$ are symmetric functions that
represent the interaction between any two neighbor sites,
and~$B_v: Q \to\mathbb{R}_{\geq 0}$ measures the influence of an "external
field" on~$v$. Then, typically, one
chooses~\mbox{$\Omega := \{ \varphi\in Q^V: w(\varphi) > 0 \}$}. The
\textit{Gibbs distribution}~$\mu$ is the probability distribution on~$\Omega$
in which the probability of an element $\varphi$ is proportional to its weight,
i.e.,~$\mu(\varphi) = w(\varphi) / Z_w$, where
$Z_w = \sum_{\phi\in\Omega} w(\phi)$.

In general, sampling from~$\mu$ is hard; in most cases, computing~$Z_w$ is
already as hard as counting~$\Omega$. This is where \textit{Glauber dynamics}
come into play. These are Markov chains~$(\varphi_t)$ that operate on~$\Omega$
an converge towards~$\mu$. In each step, they randomly select a site~$v$ (or
traverse all sites in some order) and randomly update its spin~$\varphi_t(v)$
according to a distribution~\mbox{$\nu_{\varphi_t, v}: Q\to [0, 1]$}, which is
called \textit{update rule}. A remarkable result due to Hayes and
Sinclair~\cite{hayesSinclair2007} is that for any Glauber dynamics that is
based on a \textit{local} and \textit{reversible} update rule,
$\Omega(n\log n)$ is a lower bound on the mixing time on graphs of bounded degree. Local means
that~$\nu_{\varphi, v}$ may only depend on the current spins of~$v$ and its
neighbors. An update rule is reversible if
$\mu(\varphi)\cdot \nu_{\varphi, v}(q') = \mu(\varphi')\cdot\nu_{\varphi'\kern-2pt,v}(q)$
whenever $\varphi(v)=q$ and $\varphi'$ is obtained from~$\varphi$ by changing
the spin at~$v$ to~$q'$.

Clearly, by setting~$Q = \{0, \ldots, k\}$ and with
\[
    A_{vw}(q_1, q_2) := \begin{cases}
        1 & \text{if }\lvert q_1 - q_2 \rvert \leq 1 \\
        0 & \text{otherwise}
    \end{cases}\;,\text{ \hspace{0.7cm} }B_v :\equiv 1\;,
\]
the set of feasible configurations~$\Omega$ coincides exactly with the set
of~$k$-heights, each element has the same weight, and hence, the Gibbs
distribution equals the uniform distribution~$\mathcal{U}(\Omega)$. The
up/down Markov chain~$\MM = (X_t)$ selects in each iteration uniformly a
site~$v$ and updates its spin~$X_t(v)$ depending on a coin flip and depending
on the spins of the neighbors of~$v$. Clearly, this \textit{up/down update
  rule} is local and also symmetric, hence reversible. This implies
an~$\Omega(n \log n)$ lower bound on its mixing time for all classes of graphs that we studied in Section~\ref{section:applications}, because their vertex degree is bounded.

In~\cite{blancaCaputoZongchenParisiStefankovicVigoda2022} they study
\textit{heat-bath update rules} and conditions that imply that the asymptotic
lower bound is tight, i.e., the mixing time is $\Theta(n \log n)$. An update rule 
is a heat-bath update rule, if the updated spin for~$v$ is sampled according to $\mu$
conditioned on the spins of all other vertices; more precisely,
\[
  \nu_{\varphi, v}(q) = \Prob_{\mu}[\phi(v) =
  q\;\vert\;\phi(w)=\varphi(w)\text{ for all }w\neq v]\;.
\]

It is easy to see that the heat-bath update rule coincides exactly with the
block Markov chain~$\MM_\BB$ for a family of singleton
blocks~$\BB = \left\{ \{v\}\; :\; v\in V\right\}$. In fact, the authors
in~\cite{blancaCaputoZongchenParisiStefankovicVigoda2022} also generalize
their results to so called \textit{block dynamics}, i.e., to resampling larger
blocks, just as~$\MM_\BB$ does.

\begin{theorem}[Blanca et al.,
  Theorem 1.6 in \cite{blancaCaputoZongchenParisiStefankovicVigoda2022}]
  \label{theorem:blancaEtAlOptimalMixing}
  For an arbirary spin system on a graph of maximum degree~$\Delta$, if the
  system is~$\eta$-spectrally independent and~\mbox{$b$-marginally} bounded,
  then there exists a constant~$C=C(b, \eta, \Delta)$ such that the mixing
  time of the Glauber dynamics is upper bounded
  by~\mbox{$\tau(\varepsilon) < C \cdot n \log n$},
  where~\mbox{$C =
    \left(\frac{\Delta}{b}\right)^{\mathcal{O}(1+\frac{\eta}{b})}$}.
\end{theorem}

With the next lemma we show that the spin system
of~\mbox{$k$-heights} is~\textit{$b$-marginally bounded}.
After that we comment on the other condition of
Theorem~\ref{theorem:blancaEtAlOptimalMixing}: spectral independence.

Recall the definitions of~$\Omega_B$ and~$\Omega_{B\vert X}$ in
Subsection~\ref{subsection:preliminaries_blockchain}. They naturally
generalize to spin systems other than~$k$-heights.  Let~$\mu_{B\vert X}$
denote the distribution on~$\Omega_{B\vert X}$ which is the Gibbs distribution
conditioned on the spin values of~$X$ in the set~$V\setminus B$. In the case
of \mbox{$k$-heights},~\mbox{$\mu_{B\vert X} = \mathcal{U}(\Omega_B)$}. A spin system is
called~$b$-marginally bounded, if for every~$X\in\Omega$, for
every~$B\subset V$, for every site~$v\in B$ and for any spin~$q$ for which
there exists a~$\varphi\in\Omega_{B\vert X}$ with~$\varphi(v) = q$, the
probability under~$\mu_{B\vert X}$ of seeing spin~$q$ at site~$v$ is lower
bounded by~$b$, i.e.,~$\Prob_{\mu_{B\vert X}}[\phi(v) = q] \geq b$.

In a graph $G$ we let
$B_r(v) := \{ w\in V : \operatorname{dist}(v, w) < r\}\hspace{0.1cm}$ be the
ball of radius $r$ around $v$.

\begin{lemma}
  With respect to the spin system of $k$-heights on $G$, the Gibbs
  distribution~$\mu = \mathcal{U}(\Omega)$ is $b$-marginally bounded with
  $$b := \frac{1}{(k+1)^{(\Delta - 1)^{k-1}}}\;,$$ where $\Delta$ denotes the
  maximum degree of $G$.
\end{lemma}

\begin{proof}
Fix $X\in\Omega$, a set $B \subset V$, a site $v\in B$ and spin $q\in Q$.
Define $k$-heights $L, U\in\Omega_B$ by
\begin{align*}
        L(w) &:= \max \{ q - \operatorname{dist}(v, w), 0 \}\\
        U(w) &:= \min \{ q + \operatorname{dist}(v, w), k \}
\end{align*}
for any~$w\in B$, where~$\operatorname{dist}(v, w)=\infty$ if there 
is no $v-w$ path in $B$.
Let \[f: \Omega_B\to\Omega_B, \hspace{0.5cm}Y\mapsto (Y\lor U)\land L\; ,\] where~$\lor$
\mbox{resp.}~$\land$ denotes the pointwise minimum \mbox{resp.} maximum of
two~$k$-heights. Indeed,~$f$ is well defined, as~$\Omega_B$ is closed under
these operations. Note that~$f(Y)(v) = q$, and~$Y$ and~$f(Y)$ can only differ
on sites in the ball~$B_k(v)$.
    
We will now show that~$Y\in\Omega_{B\vert X}$
implies~$f(Y)\in\Omega_{B\vert X}$. Let $(w,w')$ be an edge
with~\mbox{$w\in B$} and $w'\not\in B$, i.e.,~$w'\in\partial B$.
The claim is that $X(w') - 1 \leq f(Y)(w)\leq X(w') + 1$.
By assumption, there exists~$\varphi\in\Omega_{B\vert X}$ with~$\varphi(v) =
q$. We know~\mbox{$\varphi(w)\geq q - \operatorname{dist}(v, w)$} by
the~$k$-height property, hence, $L(w) \leq \varphi(w) \leq X(w') + 1$.
Since~$Y\in\Omega_{B\vert X}$ we have~$Y(w)\leq X(w')+1$, and we obtain
\begin{align*}
        f(Y)(w) &= ((Y\lor U)\land L)(w) \leq (L\land Y)(w) \\
        &= \max \{L(w), Y(w)\} \leq X(w') + 1\;.
\end{align*}
A similar argument shows~$f(Y)(w)\geq X(w') - 1$: first show that $Y(w)$ and
$U(w)$ and hence also~\mbox{$(Y\lor U)(w)$} are lower bounded by $X(w') - 1$, therefore,
$X(w') - 1$ is also a lower bound for the pointwise maximum
$((Y\lor U)\land L)(w) = f(Y)(w)$. This completes the proof
of~$f(Y)\in\Omega_{B\vert X}$.

Let~$W := B\setminus B_k(v)$. For an admissible
filling~$\varphi\in\Omega_{B\vert X}$, let~$\restr{\varphi}{W}\in\Omega_W$
be the restriction on~$W$ and
let~$\Omega_W^* := \{ \phi\in\Omega_W\;\vert\;\exists \varphi\in\Omega_{B\vert
  X}: \restr{\varphi}{W} = \phi\}$. By definition, every~$\phi\in\Omega_W^*$
can be extended to an admissible filling~\mbox{$\varphi\in\Omega_{B\vert X}$},
but by applying~$f$ on any such extension, it can be even extended to an
admissible filling~$\varphi\in\Omega_{B\vert X}$ with~$\varphi(v) = q$. In
other words, out of at most
\[ (k+1)^{\lvert B \cap B_k(x) \rvert}\leq(k+1)^{\lvert B_k(x) \rvert}\leq
  (k+1)^{(\Delta - 1)^{k-1}} = b^{-1}
\] ways to extend~$\phi\in\Omega_W^*$ to an admissible filling~$\varphi\in\Omega_{B\vert X}$, there exists at least one which fulfills~$\varphi(v)=q$. For a random admissible filling $\varphi\sim\mathcal{U}(\Omega_{B\vert X})$ we conclude:
$$
\Prob[\varphi(v) = q] \geq b\;.
$$
\vskip-5mm\vbox{}
\end{proof}

The second condition of Theorem~\ref{theorem:blancaEtAlOptimalMixing} requires
the spin system to be~\textit{$\eta$-spectrally independent}. This is defined
in~\cite{blancaCaputoZongchenParisiStefankovicVigoda2022} in terms of the
\textit{ALO influence matrix}. Let~$B\subset V$, $X\in \Omega$, let~$T = V \times Q$
and let~$J_{B\vert X}\in\mathbb{R}^{\mathcal{T}\times\mathcal{T}}$ be the ALO
influence matrix defined
as\[J_{B\vert X}(v, q, v', q') := \Prob_{\mu_{B\vert X}}[\varphi(v) =
  q\;\vert\;\varphi(v') = q'] - \Prob_{\mu_{B\vert X}}[\varphi(v) = q]\;.\] Now, the spin
system is said to be~$\eta$-spectrally independent if for all~$B\subset V$
and~$X\in \Omega$ the largest
eigenvalue~$\lambda_1$ of $J_{B\vert X}$ satisfies $\lambda_1 < \eta$.

Blanca et al.~\cite{blancaCaputoZongchenParisiStefankovicVigoda2022}
relate~$\eta$-spectral independence with the existence of contracting
couplings as in
Theorem~\ref{theorem:montone_coupling_classical}. Let~$U\subset V$. For a so
called \textit{pinning}~$\tau: U\to Q$, the restricted Glauber
dynamics~$(\varphi_t^\tau)$ operates on the
set~$\Omega_{V\setminus U\vert \tau}$; in each step it updates a site
in~$V\setminus U$ using the heat bath update rule w.r.t. the
distribution~$\mu_{V\setminus U\vert \tau}$. They prove the following result:

\begin{theorem}[Blanca et al., Theorem 1.10 in
  \cite{blancaCaputoZongchenParisiStefankovicVigoda2022}]
  \label{theorem:blancaEtAlSpectralIndependence}
  If for every pinning~\mbox{$\tau: U\to Q$} there is a
  coupling~\mbox{$(X_t, Y_t)$} of the restricted Glauber
  dynamics~$(\phi_t^\tau)$ and a~$\beta < 1$ such that
  \[
    \mathbb{E}[d(X_{t+1}, Y_{t+1})] < \beta\cdot d(X_{t}, Y_{t})\;,
  \]
  where~$d(X, Y) := \#\{v\in \Omega_{V\setminus U}: X(v)\neq Y(v)\}$, then the
  spin system is~\mbox{$\eta$-spectrally} independent with
  constant~\mbox{$\eta = \frac{2}{(1-\beta)n}$}.
\end{theorem}

This result can be extended to block dynamics operating
on~$\Omega_{V\setminus U\vert \tau}$ using an arbitrary family of
blocks~$\BB\subset \mathcal{P}(V\setminus U)$; see Theorem~1.11
in~\cite{blancaCaputoZongchenParisiStefankovicVigoda2022}. Under the
conditions of Theorem~\ref{theorem:main_result} the block coupling that we
constructed in Algorithm~\ref{alg:block_coupling} is a contracting coupling as
required by Theorem~\ref{theorem:blancaEtAlSpectralIndependence} in the
absence of any pinning ($U=\emptyset$). However, it seems difficult for us to
establish a contracting block coupling that is compatible with the restriction
under an arbitrary pinning~\mbox{$\tau: U\to V$}.

We leave it as an open question in which cases the spin system of~$k$-heights
satisfies~\mbox{$\eta$-spectral} independence. This would imply an optimal
mixing time of~$\Theta(n \log n)$ for $\MM_\BB$ using singleton
blocks~$\BB = \left\{ \{v\}\;:\; v\in V\right\}$.

\section{Conclusion}

In this work we studied a natural up/down Markov chain~$\MM$ operating on $k$-heights, which can also be considered as valid configurations $\sigma: V\to \{0, \cdots, k\}$ of a spin systems with hard constraints~$\lvert \sigma(v) - \sigma(w)\rvert \leq 1$ for all $\{v, w\}\in E$. We established a criterion for the boosted block Markov chain~$\MM_{\BB}$ to be rapidly mixing, which implies that~$\MM$ is rapidly mixing as well, and showed several examples of graph classes to which it applies.

\begin{question}\label{question:rapidly_mixing_for_all_k}
    Is the up/down Markov chain~$\MM$ on the graph classes discussed in Theorem~\ref{theorem:rapidly_mixing_toroidal_rectangle_graph}, Theorem~\ref{theorem:rapidly_mixing_toroidal_hexagonal_graph} and Theorem~\ref{theorem:rapidly_mixing_regular_graph} rapidly mixing for all values of $k$?
\end{question}

For some smaller values of $k$, using larger blocks~$\BB$ and more computational power, one might be able to show that Theorem~\ref{theorem:main_result} can still be applied, hence,~$\MM$ mixes rapidly. A completely affirmative answer of Question~\ref{question:rapidly_mixing_for_all_k}, i.e.~for arbitrary large values of $k$, could consist of an argument which guarantees that one can always choose the blocks large enough so that the conditions of Theorem~\ref{theorem:main_result} are satisfied. 

Intuitively, for a vertex~$v\in\partial B$ and a vertex~$w\in B$ with large distance $\operatorname{dist}(v, w)$, an increase of~$X(v)$ should not greatly affect the distribution of~$X'(w)$ when~$X'\sim\mathcal{U}(\Omega_{B\vert X})$. This intuition is captured by the concept of \textit{spatial mixing}. The spin system of $k$-heights has \textit{strong spatial mixing}, if there exist constants $\beta$ and $\alpha > 0$ so that for all $B'\subset B \subset V$ and for any pair of boundary constraints $X, Y\in\partial B$ that differ only in a single vertex $v\in\partial B$ we have \[\lVert \mathcal{U}(\Omega_{B \vert X}) - \mathcal{U}(\Omega_{B \vert Y})\rVert_{B'} < \beta\cdot\lvert B'\rvert\cdot\operatorname{exp}(-\alpha\cdot\operatorname{dist}(v, B'))\;,\] where $\lVert \mu - \mu'\rVert_{B'}$ denotes the total variation distance of the projections of $\mu, \mu'$ on $B'$. Dyer, Sinclair, Vigoda and Weitz~\cite{dyerSinclairVigodaWeitz04} have shown that spatial mixing implies a mixing time of~$\mathcal{O}(n \log n)$ if the system is~\textit{monotone}. In fact, the system of $k$-heights is monotone by Proposition~\ref{prop:stoch-dom}. Hence, the following question becomes relevant:

\begin{question}
    For which classes of graphs and for which values of $k$ does the system of $k$-heights have strong spatial mixing with respect to the uniform distribution?
\end{question}

Finally, we want to mention that~$\MM$ is clearly not rapidly mixing on any graph. The easiest example is the complete graph~$K_n$ on~$n$ vertices and $k=2$. There are three disjoint classes of~$k$-heights which form a partition of $\Omega$:
\begin{align*}
    \Omega_> &:= \{ X \in \Omega\;\vert\; \exists v\in V: X(v) = 0\} & \lvert\Omega_>\rvert &= 2^n - 1\\
    \Omega_< &:= \{ X \in \Omega\;\vert\; \exists v\in V: X(v) = 2\} & \lvert\Omega_<\rvert &= 2^n - 1\\
    \Omega_{=} &:= \{ X_1 \equiv 1\} & \lvert\Omega_{=} \rvert &= 1
\end{align*}

Every sequence of transitions between~$\Omega_<$ and~$\Omega_>$ must contain~$X_1$, and~$X_1$ is incident to exactly~$n$ transitions to each of the other two classes. Clearly, this is a \textit{bottleneck} which shows that~$\MM$ is not rapidly mixing; see \cite[Theorem 7.4]{LPW17}.

\begin{question}
    Is the up/down Markov chain~$\MM$ on $k$-heights rapidly mixing on all graphs of maximum degree bounded by~$\Delta$, for some or arbitrary values~$\Delta \geq 2$?
\end{question}

\def\sc{\scshape}
\def\tt{\ttfamily}
\bibliographystyle{plainurl}
\bibliography{bibliography}

\appendix

\newpage
\section{Block divergence on \texorpdfstring{$3$-regular}{3-regular} graphs}
\label{appendix:block_divergence_data}

\begin{table}[h] 
\caption{Block divergence for blocks of type~$1$.}
\medskip
\label{tbl:block_divergence_type1}

\fontsize{6}{7} \selectfont

\begin{minipage}{0.495\textwidth}
\begin{tabular}{rc|lll}
\toprule
$k$ & Case & $\lvert\Omega_B\rvert$ & $\lvert\Omega_{\partial B}\rvert$ & $\max E_{B, v}$ \\
\midrule
    2 & $1_3[1]$    & 15 & 27 & $\approx$ 0.727273 \\
    2 & $1_4[1]$    & 35 & 81 & $\approx$ 0.769231 \\
    2 & $1_5[1]$    & 83 & 243 & $\approx$ 0.790323 \\
    2 & $1_6[1]$    & 199 & 729 & $\approx$ 0.798658 \\
    2 & $1_7[1]$    & 479 & 2187 & $\approx$ 0.802228 \\
    2 & $1_8[1]$    & 1155 & 6561 & $\approx$ 0.803695 \\
    2 & $1_9[1]$    & 2787 & 19683 & $\approx$ 0.804306 \\
    2 & $1_{10}[1]$   & 6727 & 59049 & $\approx$ 0.804559 \\
    2 & $1_3[1,2]$   & 15 & 9 & $\approx$ 1.327273 \\
    2 & $1_4[1,2]$   & 35 & 27 & $\approx$ 1.384615 \\
    2 & $1_4[1,3]$   & 35 & 27 & $\approx$ 1.435897 \\
    2 & $1_5[1,2]$   & 83 & 81 & $\approx$ 1.415323 \\
    2 & $1_5[1,3]$   & 83 & 81 & $\approx$ 1.540323 \\
    2 & $1_6[1,2]$   & 199 & 243 & $\approx$ 1.426863 \\
    2 & $1_6[1,3]$   & 199 & 243 & $\approx$ 1.574777 \\
    2 & $1_6[1,4]$   & 199 & 243 & $\approx$ 1.552082 \\
    2 & $1_7[1,2]$   & 479 & 729 & $\approx$ 1.431858 \\
    2 & $1_7[1,3]$   & 479 & 729 & $\approx$ 1.591048 \\
    2 & $1_7[1,4]$   & 479 & 729 & $\approx$ 1.579371 \\
    2 & $1_8[1,2]$   & 1155 & 2187 & $\approx$ 1.433892 \\
    2 & $1_8[1,3]$   & 1155 & 2187 & $\approx$ 1.597510 \\
    2 & $1_8[1,4]$   & 1155 & 2187 & $\approx$ 1.592294 \\
    2 & $1_8[1,5]$   & 1155 & 2187 & $\approx$ 1.586912 \\
    2 & $1_9[1,2]$   & 2787 & 6561 & $\approx$ 1.434741 \\
    2 & $1_9[1,3]$   & 2787 & 6561 & $\approx$ 1.600246 \\
    2 & $1_9[1,4]$   & 2787 & 6561 & $\approx$ 1.597410 \\
    2 & $1_9[1,5]$   & 2787 & 6561 & $\approx$ 1.598880 \\
    2 & $1_{10}[1,2]$  & 6727 & 19683 & $\approx$ 1.435092 \\
    2 & $1_{10}[1,3]$  & 6727 & 19683 & $\approx$ 1.601372 \\
    2 & $1_{10}[1,4]$  & 6727 & 19683 & $\approx$ 1.599577 \\
    2 & $1_{10}[1,5]$  & 6727 & 19683 & $\approx$ 1.603594 \\
    2 & $1_{10}[1,6]$  & 6727 & 19683 & $\approx$ 1.600502 \\
    2 & $1_3[1,2,3]$  & 15 & 3 & $\approx$ 1.500000 \\
    2 & $1_4[1,2,3]$  & 35 & 9 & $\approx$ 1.869231 \\
    2 & $1_5[1,2,3]$  & 83 & 27 & $\approx$ 1.905707 \\
    2 & $1_5[1,2,4]$  & 83 & 27 & $\approx$ 2.051192 \\
    2 & $1_6[1,2,3]$  & 199 & 81 & $\approx$ 1.923658 \\
    2 & $1_6[1,2,4]$  & 199 & 81 & $\approx$ 2.150510 \\
    2 & $1_6[1,3,5]$  & 199 & 81 & $\approx$ 2.150510 \\
    2 & $1_7[1,2,3]$  & 479 & 243 & $\approx$ 1.930434 \\
    2 & $1_7[1,2,4]$  & 479 & 243 & $\approx$ 2.189825 \\
    2 & $1_7[1,2,5]$  & 479 & 243 & $\approx$ 2.159796 \\
    2 & $1_7[1,3,5]$  & 479 & 243 & $\approx$ 2.235299 \\
    2 & $1_8[1,2,3]$  & 1155 & 729 & $\approx$ 1.933325 \\
    2 & $1_8[1,2,4]$  & 1155 & 729 & $\approx$ 2.206921 \\
    2 & $1_8[1,2,5]$  & 1155 & 729 & $\approx$ 2.195117 \\
    2 & $1_8[1,3,5]$  & 1155 & 729 & $\approx$ 2.264221 \\
    2 & $1_8[1,3,6]$  & 1155 & 729 & $\approx$ 2.315401 \\
    2 & $1_9[1,2,3]$  & 2787 & 2187 & $\approx$ 1.935014 \\
    2 & $1_9[1,2,4]$  & 2787 & 2187 & $\approx$ 2.213945 \\
    2 & $1_9[1,2,5]$  & 2787 & 2187 & $\approx$ 2.209698 \\
    2 & $1_9[1,2,6]$  & 2787 & 2187 & $\approx$ 2.207776 \\
    2 & $1_9[1,3,5]$  & 2787 & 2187 & $\approx$ 2.277630 \\
    2 & $1_9[1,3,6]$  & 2787 & 2187 & $\approx$ 2.342016 \\
    2 & $1_9[1,4,7]$  & 2787 & 2187 & $\approx$ 2.334910 \\
    2 & $1_{10}[1,2,3]$  & 6727 & 6561 & $\approx$ 1.936494 \\
    2 & $1_{10}[1,2,4]$ & 6727 & 6561 & $\approx$ 2.216879 \\
    2 & $1_{10}[1,2,5]$ & 6727 & 6561 & $\approx$ 2.215781 \\
    2 & $1_{10}[1,2,6]$ & 6727 & 6561 & $\approx$ 2.222741 \\
    2 & $1_{10}[1,3,5]$ & 6727 & 6561 & $\approx$ 2.282993 \\
    2 & $1_{10}[1,3,6]$ & 6727 & 6561 & $\approx$ 2.354501 \\
    2 & $1_{10}[1,3,7]$ & 6727 & 6561 & $\approx$ 2.367241 \\
    2 & $1_{10}[1,4,7]$ & 6727 & 6561 & $\approx$ 2.363205 \\
\bottomrule
\end{tabular}

\end{minipage} \hfill
\begin{minipage}{0.5\textwidth}
\begin{tabular}{rc|lll}
\toprule
$k$ & Case & $\lvert\Omega_B\rvert$ & $\lvert\Omega_{\partial B}\rvert$ & $\max E_{B, v}$  \\
\midrule
    3 & $1_3[1]$    & 22 & 64 & $\approx$ 1.600000 \\
    3 & $1_4[1]$    & 54 & 256 & $\approx$ 1.789474 \\
    3 & $1_5[1]$    & 134 & 1024 & $\approx$ 1.804348 \\
    3 & $1_6[1]$    & 340 & 4096 & $\approx$ 1.831905 \\
    3 & $1_7[1]$    & 872 & 16384 & $\approx$ 1.840096 \\
    3 & $1_8[1]$    & 2254 & 65536 & $\approx$ 1.845752 \\
    3 & $1_9[1]$    & 5854 & 262144 & $\approx$ 1.847792 \\
    3 & $1_{10}[1]$   & 15250 & 1048576 & $\approx$ 1.848706 \\
    3 & $1_3[1,2]$   & 22 & 16 & $\approx$ 2.000000 \\
    3 & $1_4[1,2]$   & 54 & 64 & $\approx$ 2.142857 \\
    3 & $1_4[1,3]$   & 54 & 64 & $\approx$ 2.000000 \\
    3 & $1_5[1,2]$   & 134 & 256 & $\approx$ 2.546154 \\
    3 & $1_5[1,3]$   & 134 & 256 & $\approx$ 2.615385 \\
    3 & $1_6[1,2]$   & 340 & 1024 & $\approx$ 2.577465 \\
    3 & $1_6[1,3]$   & 340 & 1024 & $\approx$ 2.826087 \\
    3 & $1_6[1,4]$   & 340 & 1024 & $\approx$ 3.216327 \\
    3 & $1_7[1,2]$   & 872 & 4096 & $\approx$ 2.624362 \\
    3 & $1_7[1,3]$   & 872 & 4096 & $\approx$ 2.818584 \\
    3 & $1_7[1,4]$   & 872 & 4096 & $\approx$ 3.324534 \\
    3 & $1_8[1,2]$   & 2254 & 16384 & $\approx$ 2.635011 \\
    3 & $1_8[1,3]$   & 2254 & 16384 & $\approx$ 2.856485 \\
    3 & $1_8[1,4]$   & 2254 & 16384 & $\approx$ 3.403828 \\
    3 & $1_8[1,5]$   & 2254 & 16384 & $\approx$ 3.633238 \\
    3 & $1_9[1,2]$   & 5854 & 65536 & $\approx$ 2.641197 \\
    3 & $1_9[1,3]$   & 5854 & 65536 & $\approx$ 2.863505 \\
    3 & $1_9[1,4]$   & 5854 & 65536 & $\approx$ 3.426333 \\
    3 & $1_9[1,5]$   & 5854 & 65536 & $\approx$ 3.626901 \\
    3 & $1_{10}[1,2]$  & 15250 & 262144 & $\approx$ 2.643553 \\
    3 & $1_{10}[1,3]$  & 15250 & 262144 & $\approx$ 2.868846 \\
    3 & $1_{10}[1,4]$  & 15250 & 262144 & $\approx$ 3.438075 \\
    3 & $1_{10}[1,5]$  & 15250 & 262144 & $\approx$ 3.651213 \\
    3 & $1_{10}[1,6]$  & 15250 & 262144 & $\approx$ 3.620821 \\
    3 & $1_3[1,2,3]$  & 22 & 4 & $\approx$ 3.000000 \\
    3 & $1_4[1,2,3]$  & 54 & 16 & $\approx$ 2.964286 \\
    3 & $1_5[1,2,3]$  & 134 & 64 & $\approx$ 2.966667 \\
    3 & $1_5[1,2,4]$  & 134 & 64 & $\approx$ 2.982759 \\
    3 & $1_6[1,2,3]$  & 340 & 256 & $\approx$ 3.110112 \\
    3 & $1_6[1,2,4]$  & 340 & 256 & $\approx$ 3.200000 \\
    3 & $1_6[1,3,5]$  & 340 & 256 & $\approx$ 3.000000 \\
    3 & $1_7[1,2,3]$  & 872 & 1024 & $\approx$ 3.164596 \\
    3 & $1_7[1,2,4]$  & 872 & 1024 & $\approx$ 3.525963 \\
    3 & $1_7[1,2,5]$  & 872 & 1024 & $\approx$ 3.871287 \\
    3 & $1_7[1,3,5]$  & 872 & 1024 & $\approx$ 3.617647 \\
    3 & $1_8[1,2,3]$  & 2254 & 4096 & $\approx$ 3.194189 \\
    3 & $1_8[1,2,4]$  & 2254 & 4096 & $\approx$ 3.535367 \\
    3 & $1_8[1,2,5]$  & 2254 & 4096 & $\approx$ 4.042553 \\
    3 & $1_8[1,3,5]$  & 2254 & 4096 & $\approx$ 3.831169 \\
    3 & $1_8[1,3,6]$  & 2254 & 4096 & $\approx$ 4.222115 \\
    3 & $1_9[1,2,3]$  & 5854 & 16384 & $\approx$ 3.202434 \\
    3 & $1_9[1,2,4]$  & 5854 & 16384 & $\approx$ 3.578054 \\
    3 & $1_9[1,2,5]$  & 5854 & 16384 & $\approx$ 4.114039 \\
    3 & $1_9[1,2,6]$  & 5854 & 16384 & $\approx$ 4.387895 \\
    3 & $1_9[1,3,5]$  & 5854 & 16384 & $\approx$ 3.820580 \\
    3 & $1_9[1,3,6]$  & 5854 & 16384 & $\approx$ 4.332220 \\
    3 & $1_9[1,4,7]$  & 5854 & 16384 & $\approx$ 4.846281 \\
    3 & $1_{10}[1,2,3]$ & 15250 & 65536 & $\approx$ 3.206722 \\
    3 & $1_{10}[1,2,4]$ & 15250 & 65536 & $\approx$ 3.585695 \\
    3 & $1_{10}[1,2,5]$ & 15250 & 65536 & $\approx$ 4.139175 \\
    3 & $1_{10}[1,2,6]$ & 15250 & 65536 & $\approx$ 4.390973 \\
    3 & $1_{10}[1,3,5]$ & 15250 & 65536 & $\approx$ 3.859528 \\
    3 & $1_{10}[1,3,6]$ & 15250 & 65536 & $\approx$ 4.408656 \\
    3 & $1_{10}[1,3,7]$ & 15250 & 65536 & $\approx$ 4.648191 \\
    3 & $1_{10}[1,4,7]$ & 15250 & 65536 & $\approx$ 5.051252 \\
\bottomrule
\end{tabular}
\end{minipage}
\end{table}

\begin{table} 
\caption{Block divergence for blocks of type~$2$.}
\medskip
\label{tbl:block_divergence_type2}

\fontsize{6}{7} \selectfont

\begin{minipage}{0.495\textwidth}
\begin{tabular}{rc|lll}
\toprule
$k$ & Type & $\lvert\Omega_B\rvert$ & $\lvert\Omega_{\partial B}\rvert$ & $\max E_{B, v}$ \\
\midrule
    2 & $2[1]$   & 1393 & 59049 & $\approx$ 0.706599 \\
    2 & $2[2]$   & 1393 & 59049 & $\approx$ 0.782333 \\
    2 & $2[3]$   & 1393 & 59049 & $\approx$ 0.792046 \\
    2 & $2[4]$   & 1393 & 59049 & $\approx$ 0.797591 \\
    2 & $2[1,2]$  & 1393 & 19683 & $\approx$ 1.412481 \\
    2 & $2[1,3]$  & 1393 & 19683 & $\approx$ 1.411765 \\
    2 & $2[1,4]$  & 1393 & 19683 & $\approx$ 1.477771 \\
    2 & $2[1,5]$  & 1393 & 19683 & $\approx$ 1.491765 \\
    2 & $2[1,6]$  & 1393 & 19683 & $\approx$ 1.493956 \\
    2 & $2[1,7]$  & 1393 & 19683 & $\approx$ 1.486392 \\
    2 & $2[1,8]$  & 1393 & 19683 & $\approx$ 1.412481 \\
    2 & $2[2,3]$  & 1393 & 19683 & $\approx$ 1.411765 \\
    2 & $2[2,4]$  & 1393 & 19683 & $\approx$ 1.473715 \\
    2 & $2[2,5]$  & 1393 & 19683 & $\approx$ 1.541176 \\
    2 & $2[2,6]$  & 1393 & 19683 & $\approx$ 1.557661 \\
    2 & $2[2,7]$  & 1393 & 19683 & $\approx$ 1.557944 \\
    2 & $2[3,4]$  & 1393 & 19683 & $\approx$ 1.411765 \\
    2 & $2[3,5]$  & 1393 & 19683 & $\approx$ 1.540578 \\
    2 & $2[3,6]$  & 1393 & 19683 & $\approx$ 1.544729 \\
    2 & $2[4,5]$  & 1393 & 19683 & $\approx$ 1.417639 \\
    2 & $2[1,2,3]$ & 1393 & 6561 & $\approx$ 1.911765 \\
    2 & $2[1,2,4]$ & 1393 & 6561 & $\approx$ 2.065098 \\
    2 & $2[1,2,5]$ & 1393 & 6561 & $\approx$ 2.152505 \\
    2 & $2[1,2,6]$ & 1393 & 6561 & $\approx$ 2.180995 \\
    2 & $2[1,2,7]$ & 1393 & 6561 & $\approx$ 2.185449 \\
    2 & $2[1,2,8]$ & 1393 & 6561 & $\approx$ 2.115907 \\
    2 & $2[1,3,4]$ & 1393 & 6561 & $\approx$ 2.011765 \\
    2 & $2[1,3,5]$ & 1393 & 6561 & $\approx$ 2.138765 \\
    2 & $2[1,3,6]$ & 1393 & 6561 & $\approx$ 2.161765 \\
    2 & $2[1,3,7]$ & 1393 & 6561 & $\approx$ 2.176471 \\
    2 & $2[1,3,8]$ & 1393 & 6561 & $\approx$ 2.113422 \\
    2 & $2[1,4,5]$ & 1393 & 6561 & $\approx$ 2.085655 \\
    2 & $2[1,4,6]$ & 1393 & 6561 & $\approx$ 2.212074 \\
    2 & $2[1,4,7]$ & 1393 & 6561 & $\approx$ 2.219646 \\
    2 & $2[1,4,8]$ & 1393 & 6561 & $\approx$ 2.171541 \\
    2 & $2[1,5,6]$ & 1393 & 6561 & $\approx$ 2.101420 \\
    2 & $2[1,5,7]$ & 1393 & 6561 & $\approx$ 2.164148 \\
    2 & $2[1,5,8]$ & 1393 & 6561 & $\approx$ 2.171541 \\
    2 & $2[1,6,7]$ & 1393 & 6561 & $\approx$ 2.111765 \\
    2 & $2[1,6,8]$ & 1393 & 6561 & $\approx$ 2.113422 \\
    2 & $2[1,7,8]$ & 1393 & 6561 & $\approx$ 2.115907 \\
    2 & $2[2,3,4]$ & 1393 & 6561 & $\approx$ 1.911765 \\
    2 & $2[2,3,5]$ & 1393 & 6561 & $\approx$ 2.138220 \\
    2 & $2[2,3,6]$ & 1393 & 6561 & $\approx$ 2.150153 \\
    2 & $2[2,3,7]$ & 1393 & 6561 & $\approx$ 2.171258 \\
    2 & $2[2,4,5]$ & 1393 & 6561 & $\approx$ 2.085655 \\
    2 & $2[2,4,6]$ & 1393 & 6561 & $\approx$ 2.203284 \\
    2 & $2[2,4,7]$ & 1393 & 6561 & $\approx$ 2.213514 \\
    2 & $2[2,5,6]$ & 1393 & 6561 & $\approx$ 2.139037 \\
    2 & $2[2,5,7]$ & 1393 & 6561 & $\approx$ 2.213514 \\
    2 & $2[2,6,7]$ & 1393 & 6561 & $\approx$ 2.171258 \\
    2 & $2[3,4,5]$ & 1393 & 6561 & $\approx$ 1.917639 \\
    2 & $2[3,4,6]$ & 1393 & 6561 & $\approx$ 2.148560 \\
    2 & $2[3,5,6]$ & 1393 & 6561 & $\approx$ 2.148560 \\
\bottomrule  
\end{tabular}
             
\end{minipage} \hfill
\begin{minipage}{0.5\textwidth}
\begin{tabular}{rc|lll}
\toprule     
$k$ & Type & $\lvert\Omega_B\rvert$ & $\lvert\Omega_{\partial B}\rvert$ & $\max E_{B, v}$  \\
\midrule     
    3 & $2[1]$   & 3194 & 1048576 & $\approx$ 1.190238 \\
    3 & $2[2]$   & 3194 & 1048576 & $\approx$ 1.704828 \\
    3 & $2[3]$   & 3194 & 1048576 & $\approx$ 1.814364 \\
    3 & $2[4]$   & 3194 & 1048576 & $\approx$ 1.826600 \\
    3 & $2[1,2]$  & 3194 & 262144 & $\approx$ 1.834042 \\
    3 & $2[1,3]$  & 3194 & 262144 & $\approx$ 2.169631 \\
    3 & $2[1,4]$  & 3194 & 262144 & $\approx$ 2.726797 \\
    3 & $2[1,5]$  & 3194 & 262144 & $\approx$ 2.971795 \\
    3 & $2[1,6]$  & 3194 & 262144 & $\approx$ 2.961434 \\
    3 & $2[1,7]$  & 3194 & 262144 & $\approx$ 2.881866 \\
    3 & $2[1,8]$  & 3194 & 262144 & $\approx$ 2.374306 \\
    3 & $2[2,3]$  & 3194 & 262144 & $\approx$ 2.417989 \\
    3 & $2[2,4]$  & 3194 & 262144 & $\approx$ 2.702039 \\
    3 & $2[2,5]$  & 3194 & 262144 & $\approx$ 3.341018 \\
    3 & $2[2,6]$  & 3194 & 262144 & $\approx$ 3.484726 \\
    3 & $2[2,7]$  & 3194 & 262144 & $\approx$ 3.368019 \\
    3 & $2[3,4]$  & 3194 & 262144 & $\approx$ 2.604027 \\
    3 & $2[3,5]$  & 3194 & 262144 & $\approx$ 2.822416 \\
    3 & $2[3,6]$  & 3194 & 262144 & $\approx$ 3.324534 \\
    3 & $2[4,5]$  & 3194 & 262144 & $\approx$ 2.596933 \\
    3 & $2[1,2,3]$ & 3194 & 65536 & $\approx$ 2.470549 \\
    3 & $2[1,2,4]$ & 3194 & 65536 & $\approx$ 2.825476 \\
    3 & $2[1,2,5]$ & 3194 & 65536 & $\approx$ 3.323843 \\
    3 & $2[1,2,6]$ & 3194 & 65536 & $\approx$ 3.603125 \\
    3 & $2[1,2,7]$ & 3194 & 65536 & $\approx$ 3.491267 \\
    3 & $2[1,2,8]$ & 3194 & 65536 & $\approx$ 3.008458 \\
    3 & $2[1,3,4]$ & 3194 & 65536 & $\approx$ 2.860181 \\
    3 & $2[1,3,5]$ & 3194 & 65536 & $\approx$ 3.149485 \\
    3 & $2[1,3,6]$ & 3194 & 65536 & $\approx$ 3.702780 \\
    3 & $2[1,3,7]$ & 3194 & 65536 & $\approx$ 3.832265 \\
    3 & $2[1,3,8]$ & 3194 & 65536 & $\approx$ 3.308712 \\
    3 & $2[1,4,5]$ & 3194 & 65536 & $\approx$ 3.422874 \\
    3 & $2[1,4,6]$ & 3194 & 65536 & $\approx$ 3.704301 \\
    3 & $2[1,4,7]$ & 3194 & 65536 & $\approx$ 4.225000 \\
    3 & $2[1,4,8]$ & 3194 & 65536 & $\approx$ 3.865275 \\
    3 & $2[1,5,6]$ & 3194 & 65536 & $\approx$ 3.721368 \\
    3 & $2[1,5,7]$ & 3194 & 65536 & $\approx$ 3.838751 \\
    3 & $2[1,5,8]$ & 3194 & 65536 & $\approx$ 3.865275 \\
    3 & $2[1,6,7]$ & 3194 & 65536 & $\approx$ 3.557971 \\
    3 & $2[1,6,8]$ & 3194 & 65536 & $\approx$ 3.308712 \\
    3 & $2[1,7,8]$ & 3194 & 65536 & $\approx$ 3.008458 \\
    3 & $2[2,3,4]$ & 3194 & 65536 & $\approx$ 3.040346 \\
    3 & $2[2,3,5]$ & 3194 & 65536 & $\approx$ 3.373874 \\
    3 & $2[2,3,6]$ & 3194 & 65536 & $\approx$ 3.853886 \\
    3 & $2[2,3,7]$ & 3194 & 65536 & $\approx$ 4.066123 \\
    3 & $2[2,4,5]$ & 3194 & 65536 & $\approx$ 3.396416 \\
    3 & $2[2,4,6]$ & 3194 & 65536 & $\approx$ 3.676782 \\
    3 & $2[2,4,7]$ & 3194 & 65536 & $\approx$ 4.215385 \\
    3 & $2[2,5,6]$ & 3194 & 65536 & $\approx$ 4.042553 \\
    3 & $2[2,5,7]$ & 3194 & 65536 & $\approx$ 4.215385 \\
    3 & $2[2,6,7]$ & 3194 & 65536 & $\approx$ 4.066123 \\
    3 & $2[3,4,5]$ & 3194 & 65536 & $\approx$ 3.164613 \\
    3 & $2[3,4,6]$ & 3194 & 65536 & $\approx$ 3.535367 \\
    3 & $2[3,5,6]$ & 3194 & 65536 & $\approx$ 3.535367 \\
\bottomrule  
\end{tabular}
             
\end{minipage}
\end{table}  
             
\end{document}